 \documentclass[journal]{IEEEtran}

\IEEEoverridecommandlockouts
\usepackage[T1]{fontenc}
\usepackage[utf8]{inputenc}
\usepackage{comment}
\usepackage{amsmath,amssymb,amsfonts}
\usepackage{algorithmic}
\usepackage{graphicx}
\usepackage[sort, numbers]{natbib}
\usepackage{textcomp}
\usepackage[dvipsnames]{xcolor}
\usepackage{hyperref}
\usepackage{amsthm}
\usepackage{mathtools}
\usepackage{soul}

\newtheorem{theorem}{Theorem}
\newtheorem{definition}{Definition}
\newtheorem{lemma}{Lemma}

\newtheorem{assumption}{Assumption}

\title{
String Stability of a Vehicular Platoon\\ with the use of Macroscopic Information
}

\author{Marco Mirabilio, Alessio Iovine, Elena De Santis, Maria Domenica Di Benedetto, Giordano Pola
\thanks{Marco Mirabilio, Elena De Santis, Maria Domenica Di Benedetto, Giordano Pola are with the Department of Information Engineering, Computer Science and
Mathematics, Center of Excellence DEWS, University of L’Aquila. (e-mail: marco.mirabilio@graduate.univaq.it, \{elena.desantis,mariadomenica.dibenedetto,giordano.pola\}@univaq.it)
}%
\thanks{Alessio Iovine is with the Electrical Engineering and Computer Sciences (EECS) Department at UC Berkeley, Berkeley, USA. E-mail: alessio@eecs.berkeley.edu, alessio.iovine@ieee.org.
}%
}

\begin{document}

\maketitle

\begin{abstract}
    The paper investigates the possibility of using macroscopic information to improve control performance of a vehicular platoon composed of autonomous vehicles. A mesoscopic traffic modeling framework is proposed, and closed loop String Stability is achieved by exploiting Input-to-State Stability (ISS) for the analysis of the platoon. Examples of mesoscopic control laws are illustrated and simulations show the effectiveness of our approach in ensuring String Stability and improving the platoon behavior.
\end{abstract}

\vspace{-10pt}

\section{Introduction}

Interconnected autonomous vehicles have the capability to reduce stop-and-go waves propagation and traffic oscillations via the concept of String Stability  \cite{Feng2019ARC}, \cite{Monache2019CDC},  \cite{Piccoli2018TRC},  \cite{DeSantis2006},  \cite{Swaroop1996StringStability}, \cite{Ioannou1993ACC}. This concept relies on the idea that disturbances acting on an agent of the cluster should not amplify backwards in the string. Although String Stability is a property proven for the overall set of agents, in vehicular scenario traffic jamming transitions strongly depend on the amplitude of fluctuations of the leading vehicle \cite{Nagatami2000PRE}. In the case of vehicular platooning, disturbances may be due to reference speed variation, external inputs acting on each vehicle, wrong modeling, etc. To improve the platoon stability, several cases of information sharing have been considered for each leader-follower interaction, but a common characteristic is that some microscopic variables are always shared among the whole platoon, e.g. the acceleration of the platoon's leading vehicle (see \cite{Swaroop1996StringStability}) or its desired speed profile (see \cite{Besselink2017TAC}). Indeed, Vehicle-to-Infrastucture (V2I) and Vehicle-to-Vehicle (V2V) communication technologies are nowadays a reality in the smart transportation domain \cite{Uhlemann2017VTM}, and their utilization in Cooperative Adaptive Cruise Control (CACC) is widely expected  to improve traffic conditions.

This paper analyses the benefits of macroscopic information propagation in a platoon with the goal of obtaining String Stability. With respect to the current literature, the macroscopic information sharing replaces the microscopic variable that is shared among the whole platoon. The utilisation of macroscopic information for each car-following situation leads to a mesoscopic dynamical model based on a bottom-up approach. Mesoscopic models are already known in the literature \cite{li_ioannou_2004}, but usually they are used to incorporate microscopic information in macroscopic models of the traffic flow in a top-down approach, where complex Partial Differential Equations (PDEs) or discretized versions are considered  \cite{Monache2019CDC}. The goal of such a top-down modeling is to analyze the effects of the control strategies on the traffic flow. The idea of using macroscopic quantities, mainly the density, for microscopic traffic control has already been introduced in the literature, e.g. in \cite{Ioannou1997Aut}, \cite{treiber_2013_book}, \cite{Zhu2016TITS}. In \cite{treiber_2013_book}, \cite{Zhu2016TITS} and in the references therein, the focus is on simulation aspects and real data analysis. Several works are now focusing on a mesoscopic modeling for traffic control purposes. However, their target is to use mesoscopic modeling to stabilize the traffic flow considering a macroscopic description (see \cite{li_ioannou_2004},  \cite{Johansson2019ECC}, \cite{Ferrara2018ITSC}, \cite{DeSchutter2017TRC}).

The proposed approach introduces macroscopic information in a microscopic framework for improving local control strategies. It allows to obtain a model that is based on Ordinary Differential Equations (ODEs) and to provide a rigorous stability result. To this purpose, a state variable describing the macroscopic information in an aggregate formalism is defined. Such single state variable  provides the full amount of information that is needed, while avoiding state dimension explosion. The idea to exploit macroscopic information in a microscopic framework with a bottom-up approach is not new, and can be found for example in \cite{SwaroopHindra98ICCwithDensity}, \cite{Iovine2015NAHS} and \cite{Johansson2019ShockWaves}. However, no formal String Stability analysis is performed (see \cite{Iovine2015NAHS}, \cite{Johansson2019ShockWaves}), or such analysis is based on a linearisation of the model (see \cite{SwaroopHindra98ICCwithDensity}). The main advantage of using macroscopic information in a bottom-up approach consists in the possibility of performing a theoretical String Stability analysis while minimising the system complexity. 

In this paper, each follower is assumed to correctly measure the distance and speed of its leading vehicle, using for example radar and LIDAR. The leader acceleration is communicated only to its follower. To improve control performance, macroscopic information is supposed to be obtained either from the road infrastructure (V2I) or from the whole platoon (V2V). Both technologies have strengths and weaknesses. For example, V2V technology requires the macroscopic information to be propagated through the vehicles, and possibly estimated in a distributed manner by each one. On the other hand, V2I technology may provide a more reliable information at the cost of allocating several sensors along the way and computing the quantities in a centralized manner, implying a high computation request to the central computer. A thorough analysis of pros and cons of the two communication typologies is out of the scope of the present paper. We target a platoon composed by autonomous vehicles implementing CACC, but the framework is suitable for including autonomous vehicles implementing simple ACC or even human-driven vehicles as part of the platoon. 

The class of mesoscopic controllers we propose considers macroscopic information to ensure Asymptotic String Stability in a microscopic modeling of the traffic flow. The macroscopic information is viewed as a disturbance acting on the microscopic leader-follower situation, and an Input-to-State Stability (ISS) property is used to prove String Stability \cite{Sontag1989TAC}, \cite{Sontag2008}. The proposed approach is general and applied to two different spacing policies (see \cite{SWAROOP1994VSD}), a constant spacing policy and a time-varying one. In both cases, perturbation effects along the platoon are reduced. Similarly to \cite{Besselink2017TAC}, our result is obtained through an inductive method exploiting ISS. The main difference is that here ISS is ensured with respect to all the ahead leader-follower pairs. This work extends to a general class of control systems the preliminary results introduced in \cite{Mirabilio2020CDC}, that were developed for a specific case study. Simulations show the positive effects of the utilization of macroscopic information on the whole traffic platoon. An anticipatory behaviour is generated, producing a reduction in the oscillations magnitude. The filtering of the response to rapidly accelerating lead vehicles is amplified, and the traffic flow turns out to be smoother due to a better transient harmonization. 

The paper is organized as follows. Section \ref{sct:modeling} introduces the considered framework, and Section \ref{sct:control_tools} the needed control tools. Stability of a mesoscopic closed loop system is analyzed in Section \ref{sct:mesoscopic_string_stability}. Mesoscopic control laws are provided in Section \ref{sct:mesoscopic_control}, together with stability analysis. Simulations are carried out in Section  \ref{sct:simulations} and conclusions are outlined in Section \ref{sct:conclusion}.

\vspace{0.2cm}
\textbf{Notation - } $\mathbb{R}^+$ is the set of non-negative real numbers. For a vector $x\in\mathbb{R}^n$, $|x|=\sqrt{x^T x}$ is its Euclidean norm (i.e. $|x|_2$). If a different norm is used, it is indicated by a subscript (e.g. $|x|_p$ denotes the generic $p$ norm). The $\mathcal{L}_\infty$ signal norm is defined as $|x(\cdot)|_\infty^{[t_0,t]} = \sup_{t_0\leq\tau\leq t}|x(\tau)|$. The $p=\infty$ norm of a vector is denoted by $|x|_\infty = \max_{i=1...n}|x_i|$. We refer to \cite{B_khalil_2002} for the definition of Lyapunov functions, and functions $\mathcal{K}$, $\mathcal{K}_\infty$ and  $\mathcal{KL}$.

\section{Modeling and String Stability definitions}\label{sct:modeling}

\subsection{Platoon modeling}

We consider a cluster of $N+1$ vehicles, $N\in\mathbb{N}$, proceeding in the same direction on a single lane road, as in Fig.\ref{fig:reference_framework}. We make the following 
\begin{assumption}
    All the vehicles are equal, with the same length $L\in\mathbb{R}^+$ and have the common goal of maintaining a strictly positive distance among them, while keeping the same speed.
\end{assumption}
We denote with $i = 0$ the first vehicle of the platoon and with $\mathcal{I}_N = \{1,2,...,N\}$ the set of follower vehicles. The set including all the vehicles is $\mathcal{I}_N^0 = \mathcal{I}_N \cup \{0\}$.

\begin{figure}
    \centering
    \includegraphics[width = 1\columnwidth]{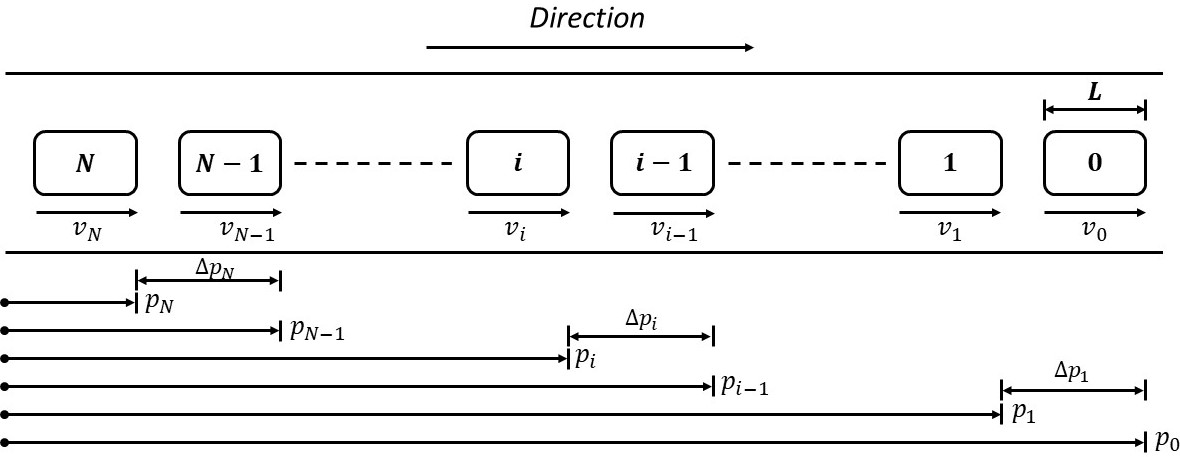}
    \caption{Reference framework.}
    \label{fig:reference_framework}
    \vspace{-10pt}
\end{figure}

\begin{color}{black}

Similarly to \cite{Besselink2017TAC}, each vehicle $i\in\mathcal{I}_N^0$ is assumed to satisfy the following generic longitudinal dynamics:
\begin{equation}\label{eq:generic_longitudinal_dynamics}
    \begin{array}{l}
        \dot{p}_i = f_p(\xi_i) \\
        \dot{\xi}_i = f_\xi(\xi_i)+g_\xi(\xi_i)u_i
    \end{array}
\end{equation}
where $p_i\in\mathbb{R}^+$ is the position of vehicle $i$, $v_i=\dot{p}_i$ ($0< v_i\leq v_{\max}, \ v_{\max}\in\mathbb{R}^+$) is its velocity and $u_i \in \mathbb{R}$ is the vehicle control input. Variable $\xi_i\in\mathbb{R}^{n-1}$ represents the remaining dynamics of the vehicle, such as actuators dynamics. The relation $v_i = \dot{p}_i = f_p(\xi_i)$, where $f_p : \mathbb{R}^{n-1} \rightarrow \mathbb{R}$, describes the interconnection between the longitudinal dynamics, velocity and position of the vehicle, with the remaining low level dynamics of the powertrain, that are described by $f_\xi(\xi_i)$ and $g_\xi(\xi_i)$, where $f_\xi : \mathbb{R}^{n-1} \rightarrow \mathbb{R}^{n-1}$ and $g_\xi : \mathbb{R}^{n-1} \rightarrow \mathbb{R}^{n-1}$. According to \cite{Swaroop1996StringStability} and \cite{Xiao2011StringStability}, the dynamics in (\ref{eq:generic_longitudinal_dynamics}) can be simplified without loss of generality in:
\begin{equation}\label{eq:longitudinal_dynamics}
    \begin{array}{l}
         \dot{p}_i = v_i, \:\:  \dot{v}_i = u_i
    \end{array}
\end{equation}
where $\xi_i=v_i$, the functions $f_p(\xi_i)=v_i$, $f_\xi(\xi_i)=0$ and $g_\xi(\xi_i)=1$; $u_i$ is the input representing both the acceleration and brake actions of vehicle $i$ ($|u_i|\leq a_{\max}, a_{\max}\in\mathbb{R}^+$). Although the double integrator model is relatively simple for robustness analysis with respect to uncertainties, it is classically used to analyse heterogeneous platoons \cite{Swaroop2019Review}. The nonlinearities and uncertainties due to the vehicle powertrain dynamics, that are strictly linked with the mechanical components, require the use of techniques such as inversion model and feedback linearisation that allow reducing the vehicle dynamic model complexity, as shown in \cite{FengGao2020HInfController} and \cite{FengGao2020RobustDistrConsensus}. For this reason, in the present work we consider the nominal model in (\ref{eq:longitudinal_dynamics}) to illustrate our approach. 
\end{color}

To provide a global description of the platoon, we adopt the leader-follower model that describes the inter-vehicular interaction  (see \cite{Ploeg2014TCST}). We define the state of each vehicle $i\in\mathcal{I}_N^0$ as
\begin{equation}\label{eq:state_vector_x}
    x_i = \left[ p_i \ \ v_i \right]^T
\end{equation}
and the state of each car-following situation among the leading vehicle $i-1$ and the following one $i$ as
\begin{equation}\label{eq:state_vector_chi}
    \chi_i = x_i-x_{i-1} 
    = \left[
        \begin{array}{c}
            \Delta p_i    \\
            \Delta v_i
        \end{array}
    \right] =  \left[
        \begin{array}{c}
            p_i-p_{i-1} \\
            v_i-v_{i-1}
        \end{array}
    \right].
\end{equation}
Positions, speed and acceleration of each leading vehicle are supposed to be perfectly known, either measured or communicated to the following one. Consequently, the obtained dynamical model is:
\begin{equation}\label{eq:car_following_dynamics_i}
   \dot{\chi}_i = \left[
    \begin{array}{c}
      \Delta \dot{p}_i    \\
      \Delta \dot{v}_i
    \end{array}
    \right] =  \left[
    \begin{array}{c}
         \Delta {v}_i \\
         u_i-u_{i-1}
    \end{array}
   \right], \quad i\in\mathcal{I}_N,
\end{equation}
or, in short,
\begin{align}\label{eq:car_following_dynamics_i_short}
    \dot{\chi}_i  = f(\chi_i,u_i,u_{i-1}), \quad i\in\mathcal{I}_N.
\end{align}
Equation (\ref{eq:car_following_dynamics_i_short}) shows that the car-following dynamics are independent of the position in the platoon. Indeed, the function $f:\mathbb{R}^2 \times \mathbb{R} \times \mathbb{R} \rightarrow \mathbb{R}^2$ is the same for each vehicle of the set $\mathcal{I}_N$. To derive the dynamics of the first vehicle $i=0$ of the platoon in the same form of (\ref{eq:car_following_dynamics_i}), a non-autonomous non-communicating  virtual leader $i=-1$ is considered to precede the set of vehicles, with dynamical model
\begin{equation}\label{eq:virtual_leader}
    \dot{x}_{-1} = \left[
    \begin{array}{c}
         \dot{p}_{-1}  \\
         \dot{v}_{-1} 
    \end{array}
    \right] = \left[
    \begin{array}{c}
         v_{-1}  \\
         u_{-1}
    \end{array}
    \right].
\end{equation}
Then, the car-following dynamics with respect to vehicle $i=0$ can be described by:
\begin{equation}\label{eq:car_following_dynamics_0}
   \dot{\chi}_0 = \left[
    \begin{array}{c}
      \Delta \dot{p}_0    \\
      \Delta \dot{v}_0
    \end{array}
    \right] =  \left[
    \begin{array}{c}
         \Delta {v}_0 \\
         u_0-u_{-1}
    \end{array}
   \right].
\end{equation}
It follows that dynamics in (\ref{eq:car_following_dynamics_i_short}) is valid $\forall \  i\in\mathcal{I}_N^0$. Since we consider $i=-1$ to represent a virtual vehicle, $p_{-1}(t) = \int_0^t v_{-1}(\tau)d\tau, \ t\geq 0,$ is a dummy state. Moreover, we assume $\Delta p_0(t) = -\Delta\Bar{p}, \ \forall t\geq 0$, where $\Delta\Bar{p} > 0$ is the constant desired inter-vehicular distance. Since $p_i<p_{i-1}$, the desired distance has to be negative.
In accordance to \cite{Ploeg2014TCST}, \cite{ZhengBorrelli2017TCST}, \cite{Zheng2014ITSC}, we consider the  widely accepted hypothesis of constant speed for the virtual leader $i=-1$, that precedes the entire cluster. Then, we have
\begin{equation}
    p_{-1}(t) = \Bar{v}\cdot t, \ v_{-1}(t) = \Bar{v}, \ u_{-1}(t) = 0, \ \forall \ t\geq 0,
\end{equation}
where $\Bar{v}>0$ is a constant speed. 
For vehicular platoons, the constant speed assumption defines the equilibrium point for all the vehicles in the cluster. Consequently, when all $i\in\mathcal{I}_N^0$ have equal speed and are at the same desired distance $\Delta\Bar{p}$, the equilibrium point for the $i$-th system of dynamics (\ref{eq:car_following_dynamics_i_short}) is
\begin{equation}\label{eq:chi_equilibrium}
   \chi_{e,i} = \Bar{\chi} = [ -\Delta\Bar{p} \ \ \ 0 \ ]^T.
\end{equation}
From a global point of view, we define the platoon lumped state and the corresponding equilibrium point for $u_{-1} = 0$ as
\begin{equation}\label{eq:platoon_state_and_equilibrium}
    \chi=[\chi_0^T \ \chi_1^T \ ... \ \chi_N^T]^T, \quad \chi_e = [ \Bar{\chi}^T \ \Bar{\chi}^T \ ... \ \Bar{\chi}^T]^T.
\end{equation}
%

\subsection{String Stability definitions}\label{sct:string_stability}

Let the model in (\ref{eq:car_following_dynamics_i}) describe the $i$-th vehicle of the platoon, and its equilibrium be (\ref{eq:chi_equilibrium}). The control input $u_i$ is generated by the following dynamic controller
\begin{equation}\label{eq:dynamic_control_law}
    \begin{cases}
        \dot{\rho}_i = \omega_i( \rho_i , \chi  ) \\
        u_i = h_i( \rho_i , \chi , \chi_{e,i} , u_{i-1} )
    \end{cases}
\end{equation}
where $\rho_i \in \mathbb{R}^r$ is the state vector of the dynamic controller with dimension $r \geq 1$; $\omega_i : \mathbb{R}^r \times \mathbb{R}^{2N} \rightarrow \mathbb{R}^r$ is the vector field describing the evolution of the controller state; $h_i : \mathbb{R}^r \times \mathbb{R}^{2N} \times \mathbb{R}^2 \times \mathbb{R} \rightarrow \mathbb{R}$ is the output function which outcome is the control input $u_i$. The dynamical system (\ref{eq:dynamic_control_law}) takes as inputs $\chi, \ \chi_{e,i}, \ u_{i-1}$; these quantities are supposed to be shared through the adopted communication technology either via V2I or V2V. The resulting closed loop system is denoted in the sequel by $P_{cl}$, where the state vector of vehicle $i$ and the corresponding equilibrium point are, respectively,
\begin{equation}\label{eq:chi_extended_and_equilibrium}
    \Hat{\chi}_i = [ \: \chi_i^T \:\: \rho_i^T \: ]^T, \:\: \Hat{\chi}_{e,i} = [ \: \Bar{\chi}^T \:\: 0_{r}^T \: ]^T, \: \forall \: i \in \mathcal{I}_N^0,
\end{equation}
with $0_r\in\mathbb{R}^r$ a null column vector. 

We now recall the notions of String Stability and Asymptotic String Stability from \cite{Swaroop1996StringStability} and \cite{Besselink2017TAC}.
\begin{definition}{(String Stability)}\label{def:TSS_string_stability}
    The equilibrium $\Hat{\chi}_{e,i}, \: i\in\mathcal{I}_N^0$, of $P_{cl}$ is said to be String Stable if, for any $\epsilon>0$, there exists $\delta>0$ such that, for all $N\in\mathbb{N}$,
    \begin{equation}
        \max_{i\in\mathcal{I}_N^0}|\Hat{\chi}_i(0)-\Hat{\chi}_{e,i}|<\delta \Rightarrow \max_{i\in\mathcal{I}_N^0}|\Hat{\chi}_i(t)-\Hat{\chi}_{e,i}|<\epsilon, \ \forall \ t\geq 0.
    \end{equation}
\end{definition}
\begin{definition}{(Asymptotic String Stability)}\label{def:ATSS_string_stability}
    The equilibrium $\Hat{\chi}_{e,i}, \ i\in\mathcal{I}_N^0$, of $P_{cl}$ is said to be Asymptotically String Stable if it is String Stable and, for all $N\in\mathbb{N}$,
    \begin{equation}
        \lim_{t\rightarrow\infty}|\Hat{\chi}_i(t)-\Hat{\chi}_{e,i}|=0, \quad \forall \ i\in\mathcal{I}_N^0.
    \end{equation}
\end{definition}
%

\section{Control tools}\label{sct:control_tools}

The goal of this paper is to design controllers as in (\ref{eq:dynamic_control_law}) that adopt mesoscopic quantities for ensuring String Stability of $P_{cl}$. To this purpose, proper spacing policies and a function describing macroscopic information are now introduced. 

\subsection{Spacing policies}\label{sct:policies}

Several spacing policies have been introduced in the literature (e.g. \cite{Swaroop1999DynSystCont}, \cite{Hedrick2004strings}). We consider two of the most investigated ones: a constant spacing policy and a variable time spacing policy. Our goal is to slightly modify the reference trajectories for a better transient harmonization when traffic conditions vary, without changing the platoon equilibrium in (\ref{eq:platoon_state_and_equilibrium}) in steady-state.
    
\subsubsection{Constant spacing policy} 

A constant spacing strategy consists in tracking a constant desired inter-vehicular distance. The reference distance is defined as
\begin{equation}\label{eq:constant_spacing_policy}
    \Delta p_i^r(t) = -\Delta\Bar{p}, \:\:\forall\: i\in\mathcal{I}_N^0
\end{equation}
where $\Delta \Bar{p}>0$ is a constant value. Usually, the main drawback of this policy is that it does not guarantee String Stability if the control input of vehicle $i$ exploits only measurements with respect to its preceding vehicle $i-1$ (see \cite{Hedrick2004strings}). 

\subsubsection{Variable time spacing policy} 

A variable time spacing policy consists in tracking a variable inter-vehicular desired distance. This approach  allows  for  a  lower inter-vehicular  spacing  at  steady-state  conditions  (see \cite{Yanakiev1995CDC} and \cite{Zhang2005CDC}). We define a \textit{mesoscopic} time varying trajectory for the distance policy $\Delta p^r_i$ of the $i$-th vehicle with respect to its leader $i-1$:
\begin{equation}\label{eq:mesoscopic_spacing_policy}
    \Delta p_{i}^r(t) = -\Delta\Bar{p} - \rho^M_i(t), \ t\geq 0
\end{equation}
where $\Delta\Bar{p}>0$ is a constant inter-vehicular distance and $\rho^M_i(t)$ is a function describing macroscopic information. 

\subsection{Macroscopic information}\label{sct:macroscopic_info}

In this subsection, on the basis of the analysis carried out in \cite{treiber_2013_book}, we define appropriate macroscopic functions taking into account microscopic distance and speed variance. 

Given the generic vehicle $i\in\mathcal{I}_N$, let $\mu_{\Delta p,i}$ and $\sigma^2_{\Delta p,i}$ be the inter-vehicular distance mean and variance computed from vehicle $0$ to vehicle $i$, respectively:
\begin{equation}\label{eq:meanvar_distance}
    \mu_{\Delta p,i} = \frac{1}{i+1}\sum_{j = 0}^{i}\Delta p_j , \quad
    \sigma_{\Delta p,i}^2 = \frac{1}{i+1}\sum_{j = 0}^{i}(\Delta p_j-\mu_{\Delta p_i})^2.
\end{equation}
Let $\mu_{\Delta v,i}$ and $\sigma^2_{\Delta v,i}$ be the velocity tracking error mean and variance computed from vehicle $0$ to vehicle $i$, respectively:
\begin{equation}\label{eq:meanvar_speederror}
    \mu_{\Delta v,i} = \frac{1}{i+1}\sum_{j = 0}^{i}\Delta v_j, \quad
    \sigma_{\Delta v,i}^2 = \frac{1}{i+1}\sum_{j = 0}^{i}(\Delta v_j-\mu_{\Delta v_i})^2 .
\end{equation}
We denote by 
\begin{align}\label{eq:fmacro_pv}
    \psi^i_{\Delta p} : \mathbb{R}\times\mathbb{R}^+ \rightarrow\mathbb{R}, \:\: \psi^i_{\Delta v} : \mathbb{R}\times\mathbb{R}^+ \rightarrow\mathbb{R},
\end{align}
the distance macroscopic function and the speed tracking error macroscopic function, respectively.

The scope of functions $\psi^i_{\Delta p}$ and $\psi^i_{\Delta v}$ is to provide the $i$-th vehicle with the macroscopic information about the platoon system in an aggregate way. Since these functions are based on the variances of the microscopic and speed differences, they represent the state of the vehicles ahead, indicating the existence of a transient or a steady state phase. Consequently, the case $\psi^i_{\Delta p} = \psi^i_{\Delta v} = 0$ means that the vehicles up to the $i$-$th$ one have no oscillations. Instead of considering the whole set of leader-follower situations, a complexity reduction of the considered interconnection framework is obtained without reducing the level of available information. We embed the macroscopic information given by $\psi^i_{\Delta p}$ and $\psi^i_{\Delta v}$ in the macroscopic function $\rho_i$, the evolution of which is given by the controller dynamics in (\ref{eq:dynamic_control_law}). In order to describe the evolution of the controller state and of its interconnection with the macroscopic information, we propose the following dynamical system:
\begin{equation}\label{eq:rho_dynamic_system}
    \begin{cases}
        \dot{\rho}_i = \Lambda \rho_i + G_{\rho} [ \: \psi^{i-1}_{\Delta p} \:\: \psi^{i-1}_{\Delta v} \: ]^T \\
        \rho_i(0) = \rho_0
    \end{cases}
\end{equation}
where $\Lambda\in\mathbb{R}^{r \times r}$ is an upper triangular matrix which diagonal elements are strictly negative, so that the state trajectory of (\ref{eq:rho_dynamic_system}) is asymptotically stable and $\rho_i \rightarrow 0$ as $\psi^i_{\Delta p} \rightarrow 0$ and $\psi^i_{\Delta v} \rightarrow  0$; $G_{\rho}\in\mathbb{R}^{r \times 2}$ is the input matrix; $\rho_0\in\mathbb{R}^r$ is the initial condition. The controller response  with respect to variations of the macroscopic variables depends on appropriate choice of the elements of $\Lambda$. Also, it is possible to choose different weights for the considered macroscopic information through the input matrix $G_\rho$. The superscript $i-1$ in the input vector $[\: \psi^{i-1}_{\Delta p} \:\: \psi^{i-1}_{\Delta v} \:]^T$ means that we consider the macroscopic information calculated up to the preceding vehicle. We also set $\psi^{-1}_{\Delta p} = \psi^{i-1}_{\Delta v} = 0$ for $i=0$. The macroscopic function $\rho_i^M$ in the variable spacing policy (\ref{eq:mesoscopic_spacing_policy}) is defined as a linear combination of the components of $\rho_i$. The main advantage in considering $\rho^M_i$ consists in incorporating the whole macroscopic information about the platoon, thus avoiding complexity calculation explosion due to processing non aggregate state information by the control law.

\section{Mesoscopic String Stability}\label{sct:mesoscopic_string_stability}

In this section, we analyse String Stability of the closed loop system, by referring to the extended state $\hat{\chi}_i $ and its equilibrium point $\hat{\chi}_{e,i}$ in (\ref{eq:chi_extended_and_equilibrium}). Let $\Tilde{\chi}_i = \hat{\chi}_i - \hat{\chi}_{e,i}$ and $g_{cl,0}(\Tilde{\chi}_{-1}) = 0$. The closed loop dynamics of (\ref{eq:chi_extended_and_equilibrium}), where each control input is the output signal of the dynamic system in (\ref{eq:dynamic_control_law}), is described by
\begin{align}
    \dot{\Tilde{\chi}}_0 &= f_{cl}(\Tilde{\chi}_0), \ i=0, \label{eq:closed_loop_platoon_short_0}\\
    \dot{\Tilde{\chi}}_i &= f_{cl}(\Tilde{\chi}_i) + g_{cl,i}(\Tilde{\chi}_{i-1},\Tilde{\chi}_{i-2},...,\Tilde{\chi}_0), \ \forall \  i\in\mathcal{I}_N, \label{eq:closed_loop_platoon_short_i}
\end{align}  
where $f_{cl}:\mathbb{R}^{2+r} \rightarrow \mathbb{R}^{2+r}$ is the vector field describing the evolution dynamics of each isolated subsystem, and $g_{cl,i}:\underbrace{\mathbb{R}^{2+r}\times\cdot\cdot\cdot\times\mathbb{R}^{2+r}}_{i \text{ times}} \rightarrow \mathbb{R}^{2+r}$ is the interconnection term. Moreover, $f_{cl}(0) = 0$ and $g_{cl,i}(0,0,...,0) = 0$.
We assume that the virtual leader $i=-1$ has a constant speed $v_{-1} = \Bar{v}>0, \ u_{-1} = 0$. The assumption of $\Delta p_0(t) = -\Delta\Bar{p}, \ \forall t\geq 0$ is considered. 

\begin{theorem}\label{thm:StringStability}
    Consider the autonomous system in (\ref{eq:closed_loop_platoon_short_0}) and (\ref{eq:closed_loop_platoon_short_i}). If there exist functions $\beta$ of class $\mathcal{KL}$ and $\gamma$ of class $\mathcal{K}_\infty$, such that for each $i\in\mathcal{I}_N^0$
    \begin{equation}\label{eq:ISS_platoon_condition}
        |\Tilde{\chi}_i(t)| \leq \beta( |\Tilde{\chi}_i(0)| , t ) + \gamma\left( \max_{j=0,...,i-1}|\Tilde{\chi}_j(\cdot)|_{\infty}^{[0,t]} \right)
    \end{equation}
    $\forall \ t \geq 0$ and $\gamma(s) \leq \Tilde{\gamma}s$, for $s \geq 0$  and $\Tilde{\gamma}\in(0,1)$, then the interconnected system is String Stable. Moreover, if the origin of each isolated subsystem is exponentially stable and each interconnection term $g_{cl,i}, \ \forall i\in\mathcal{I}_N^0$, is bounded by 
    \begin{equation}\label{eq:gcl_sum_bound_condition}
        |g_{cl,i}(\Tilde{\chi}_{i-1},\Tilde{\chi}_{i-2},...,\Tilde{\chi}_0)| \leq \sum_{j=0}^{i-1} k_{ij}|\Tilde{\chi}_j|
    \end{equation}
    with constants $k_{ij} \in \mathbb{R}^+, \ \forall \ j\in\mathcal{I}_N^0$, then the autonomous system is Asymptotically String Stable.
\end{theorem}

\begin{proof}
    The first part of the proof is based on the forward recursive application of the Input-to-State Stability (ISS) property in (\ref{eq:ISS_platoon_condition}) through an inductive method.
    Since the assumption $\gamma(s) \leq \Tilde{\gamma}s$ holds, then dynamics $\Tilde{\chi}_{i=\{0,1\}}$ trivially verifies
    \begin{align}
        |\Tilde{\chi}_0(t)| &\leq \beta(|\Tilde{\chi}_0(0)|,t), \ \forall \ t \geq 0, \\
        |\Tilde{\chi}_1(t)| &\leq \beta(|\Tilde{\chi}_1(0)|,t)+\Tilde{\gamma}|\Tilde{\chi}_0(\cdot)|^{[0,t]}_\infty, \ \forall \ t \geq 0,
    \end{align}
    where $|\Tilde{\chi}_0(\cdot)|^{[0,t]}_\infty \leq \beta(|\Tilde{\chi}_0(0)|,0)$. Defining $|\Tilde{\chi}_M(0)| = \max\{ |\Tilde{\chi}_0(0)| , |\Tilde{\chi}_1(0)| \}$, then for $i=0$ and $i=1$:
    \begin{align}
        |\Tilde{\chi}_0(t)| &\leq \beta(|\Tilde{\chi}_M(0)|,0), \ \forall \ t \geq 0, \\
        |\Tilde{\chi}_1(t)| &\leq \beta(|\Tilde{\chi}_M(0)|,0)(1+\Tilde{\gamma}), \ \forall \ t \geq 0.
    \end{align}
    Since $\Tilde{\gamma} \in \mathbb{R}^+$, then 
    \begin{align}
        |\Tilde{\chi}_{i=0,1}(t)| \leq \beta(|\Tilde{\chi}_M(0)|,0)(1+\Tilde{\gamma}), \ \forall \ t \geq 0.
    \end{align}
    Dynamics $\Tilde{\chi}_{i=2}$ verifies:
    \begin{align}
        \nonumber |\Tilde{\chi}_2(t)| &\leq \beta(|\Tilde{\chi}_2(0)|,t) \\
        & \quad +\Tilde{\gamma}\max_{j=0,1}|\Tilde{\chi}_j(\cdot)|^{[0,t]}_\infty, \ \forall \ t \geq 0.
    \end{align}
    By defining $|\Tilde{\chi}_M'(0)| = \max_{j=0,1,2}\{ |\Tilde{\chi}_j(0)| \}$, then
    \begin{equation}
        |\Tilde{\chi}_2(t)| \leq \beta(|\Tilde{\chi}_M'(0)|,0)(1+\Tilde{\gamma}+\Tilde{\gamma}^2), \ \forall \ t \geq 0,
    \end{equation}
    where the bound holds also for $i=0$ and $i=1$. By recursively applying these steps, and since the assumption $\Tilde{\gamma}\in(0,1)$ holds, then for each $i\in\mathcal{I}_N^0$ the following inequality is verified:
    \begin{align}\label{eq:chi_beta_bound}
        \nonumber |\Tilde{\chi}_i(t)| &\leq \beta\left( \max_{j=0,...,i}|\Tilde{\chi}_j(0)|,0 \right)\sum_{j=0}^i \Tilde{\gamma}^j \\
        \nonumber &\leq \beta\left( \max_{j=0,...,i}|\Tilde{\chi}_j(0)|,0 \right)\sum_{j=0}^\infty \Tilde{\gamma}^j \\
        &\leq \frac{1}{1-\Tilde{\gamma}}\beta\left( \max_{j=0,...,i}|\Tilde{\chi}_j(0)|,0 \right), \ \forall \ t \geq 0.
    \end{align}
    Then, by applying $\max\{\cdot\}$ operator to the first and last term of (\ref{eq:chi_beta_bound}) we get:
    \begin{equation}\label{eq:string_stability_result}
        \max_{i\in\mathcal{I}_N^0}|\Tilde{\chi}_i(t)| \leq \frac{1}{1-\Tilde{\gamma}}\beta\left( \max_{i\in\mathcal{I}_N^0}|\Tilde{\chi}_i(0)|,0 \right), \ \forall \ t \geq 0.
    \end{equation}
    Let us define $\omega(s) = \beta(s,0), \ s \geq 0$. By definition of $\mathcal{KL}$ functions, $\omega$ is $\mathcal{K}_\infty$ and hence invertible. Since (\ref{eq:string_stability_result}) holds for any $t \geq 0$, then 
    \begin{equation}\label{eq:chi_delta_bound}
        \delta = \omega^{-1}((1-\Tilde{\gamma})\epsilon), \quad \forall \ \epsilon \geq 0.
    \end{equation}
    The value of $\delta$ in (\ref{eq:chi_delta_bound}) does not depend on the system dimension. From (\ref{eq:chi_beta_bound}), (\ref{eq:string_stability_result}) and (\ref{eq:chi_delta_bound}), String Stability is ensured according to Definition \ref{def:TSS_string_stability}.
    
    We focus now on the possibility to ensure Asymptotic String Stability. This second part is based on a composition of Lyapunov functions (see \cite{B_khalil_2002}).
    By the converse Lyapunov theorem, there exists a function $W:\mathbb{R}^3 \rightarrow \mathbb{R}^+$ and constants $\underline{\alpha},\Bar{\alpha},\alpha,\alpha' > 0$ such that
    \begin{align}
        \underline{\alpha}|\Tilde{\chi}_i|^2 \leq W(\Tilde{\chi}_i) \leq \Bar{\alpha}|\Tilde{\chi}_i|^2 \\
        \frac{\partial W(\Tilde{\chi}_i) }{\partial \Tilde{\chi}_i}f_{cl}(\Tilde{\chi}_i) \leq -\alpha|\Tilde{\chi}_i|^2\\
        \left| \frac{\partial W(\Tilde{\chi}_i) }{\partial \Tilde{\chi}_i} \right| \leq \alpha'|\Tilde{\chi}_i|
    \end{align}
    In the following, for the sake of notational simplicity, we denote $W_i = W(\Tilde{\chi}_i)$. By computing the time derivative of the Lyapunov function $W_i$ with respect to the system with non-zero interconnected term, we obtain:
    \begin{align}
        \nonumber \dot{W}_i &= \frac{\partial W}{\partial \Tilde{\chi}_i}(f_{cl}(\Tilde{\chi}_i) + g_{cl,i}(\Tilde{\chi}_{i-1},...,\Tilde{\chi}_0)) \\
        \nonumber &\leq -\alpha|\Tilde{\chi}_i|^2 + \left| \frac{\partial W}{\partial \Tilde{\chi}_i} \right| |g_{cl,i}(\Tilde{\chi}_{i-1},...,\Tilde{\chi}_0))| \\
        &\leq -\alpha|\Tilde{\chi}_i|^2 + \alpha' |\Tilde{\chi}_i| |g_{cl,i}(\Tilde{\chi}_{i-1},...,\Tilde{\chi}_0)|
    \end{align}
    If condition in (\ref{eq:gcl_sum_bound_condition}) is verified, then 
    \begin{align}
        \dot{W}_i &\leq -\alpha|\Tilde{\chi}_i|^2 + \alpha' |\Tilde{\chi}_i|\sum_{j=0}^{i-1}k_{ij}|\Tilde{\chi}_j|.
    \end{align}
    Let us introduce $\hat{\chi}$ and $\hat{\chi}_e$ as the extended lumped state of the platoon and the extended equilibrium point, respectively, that are defined similarly as in (\ref{eq:platoon_state_and_equilibrium}). Let us consider $\Tilde{\chi} = \hat{\chi} - \hat{\chi}_e$ and the parameters $d_{c,i} > 0$. Then, we define the composite function $W_c$:
    \begin{equation}\label{eq:Lyapunov_composite}
        W_c(\Tilde{\chi}) = \sum_{i=0}^N d_{c,i} W(\Tilde{\chi}_i).
    \end{equation}
    It satisfies
    \begin{align}
        &\underline{\alpha}_c|\Tilde{\chi}|^2 \leq W_c(\Tilde{\chi}) \leq \Bar{\alpha}_c|\Tilde{\chi}|^2, \\
        \underline{\alpha}_c &= \min_{i\in\mathcal{I}_N^0}\{ d_{c,i} \}\underline{\alpha}, \:\: \Bar{\alpha}_c = \max_{i\in\mathcal{I}_N^0}\{ d_{c,i} \}\Bar{\alpha}.
    \end{align}
    The time derivative of the composite function in (\ref{eq:Lyapunov_composite}) is
    \begin{align}\label{eq:dot_Lyapunov_composite}
        \dot{W}_c(\Tilde{\chi}) \leq \sum_{i=0}^N d_i \left[ -\alpha|\Tilde{\chi}_i|^2 + \alpha' |\Tilde{\chi}_i|\sum_{j=0}^{i-1} k_{ij}|\Tilde{\chi}_j| \right].
    \end{align}
    We define the operator $\phi:\mathbb{R}^{2N+1} \rightarrow \mathbb{R}^{N+1}$ as
    \begin{equation}
        \phi(\Tilde{\chi}) = [ |\Tilde{\chi}_0| \ |\Tilde{\chi}_1| \ ... \ |\Tilde{\chi}_N| ]^T.
    \end{equation}
    Then, equation (\ref{eq:dot_Lyapunov_composite}) can be written as
    \begin{equation}\label{eq:dot_Lyapunov_composite_phi_bound}
        \dot{W}_c(\Tilde{\chi}) \leq -\frac{1}{2}\phi(\Tilde{\chi})^T(DS + S^T D)\phi(\Tilde{\chi}),
    \end{equation}
    where 
    \begin{equation}
        D = diag(d_{c,0},d_{c,1},...,d_{c,N})
    \end{equation}
    and $S$ is an $N \times N$ matrix whose elements are
    \begin{equation}
        s_{ij} = 
        \begin{cases}
            \alpha, \:\:\: & \text{if } i=j \\
            -\alpha' k_j, \:\:\: & \text{if } i<j \\
            0, \:\:\: & \text{if } i>j
        \end{cases}
    \end{equation}
    Since $\alpha > 0$, each leading principal minor of $S$ is positive and hence it is an $M-$matrix. By \cite[Lemma~9.7]{B_khalil_2002} there exists a matrix $D$ such that $DS+S^T D > 0$. Consequently, $\dot{W}_c$ in (\ref{eq:dot_Lyapunov_composite_phi_bound}) is negative definite. It follows that $W_c$ in (\ref{eq:Lyapunov_composite}) is a Lyapunov function for the overall autonomous system described by (\ref{eq:closed_loop_platoon_short_0}) and (\ref{eq:closed_loop_platoon_short_i}). Therefore, there exists a $\mathcal{KL}$ function $\beta_c : \mathbb{R}^+ \times \mathbb{R}^+ \rightarrow \mathbb{R}^+$ such that
    \begin{equation}\label{eq:chi_beta_composite_bound}
        |\Tilde{\chi}(t)| \leq \beta_c(|\chi(0)|,t), \:\: \forall \ t \geq 0.
    \end{equation}
    Condition in (\ref{eq:chi_beta_composite_bound}) ensures the asymptotic stability:
    \begin{equation}\label{eq:chi_asymptotic_limit}
        \lim_{t \rightarrow \infty} |\Tilde{\chi}_i(t)| = 0, \:\: \forall \ i \in \mathcal{I}_N^0.
    \end{equation}
    The platoon system is proved to be String Stable by (\ref{eq:string_stability_result}) and (\ref{eq:chi_delta_bound}). Consequently, for each $i\in\mathcal{I}_N^0$ the state evolution $|\Tilde{\chi}_i|$ is constrained by a bound that is independent from the system dimension. Furthermore, from (\ref{eq:chi_asymptotic_limit}) Asymptotic String Stability is ensured according to Definition \ref{def:ATSS_string_stability}.
\end{proof}
The consequences of the interconnection term $g_{cl,i}$ in (\ref{eq:closed_loop_platoon_short_i}) acting on the closed loop system are described by the function $\gamma$ in (\ref{eq:ISS_platoon_condition}).  As it is shown in (\ref{eq:string_stability_result}), the higher the contribution of the macroscopic information, the greater the perturbation acting on each leader-follower situation. However, this perturbation acting on each follower results in a cascade effect that ensures Asymptotic String Stability of the whole platoon. 

\section{Mesoscopic Control Laws}\label{sct:mesoscopic_control}

In this section, we show the generality of our approach by introducing two control laws that consider mesoscopic quantities for a single car-following situation. String Stability and Asymptotic String Stability as in Definitions \ref{def:TSS_string_stability} and \ref{def:ATSS_string_stability} are proved when the proposed control laws are implemented for each leader-follower situation along the platoon. The first control law adopts the constant spacing policy in (\ref{eq:constant_spacing_policy}), while the second one implements the variable spacing policy in (\ref{eq:mesoscopic_spacing_policy}). Both  consider the function $\rho_i$ describing macroscopic information in (\ref{eq:rho_dynamic_system}). Each vehicle is modeled according to dynamics (\ref{eq:longitudinal_dynamics}) and each car-following situation according to $\chi_i$ in (\ref{eq:car_following_dynamics_i}) and (\ref{eq:car_following_dynamics_0}). 

\subsection{Macroscopic functions properties}

In order to exploit macroscopic information, functions $\psi^i_{\Delta p}, \ \psi^i_{\Delta v}$ in (\ref{eq:fmacro_pv}) are computed with respect to the mean and variance in (\ref{eq:meanvar_distance}) and (\ref{eq:meanvar_speederror}): $\psi^i_{\Delta p}(\mu_{\Delta p,i},\sigma^2_{\Delta p,i})$, $\psi^i_{\Delta v}(\mu_{\Delta v,i},\sigma^2_{\Delta v,i})$. In the sequel, for simplicity we omit the arguments. Moreover, by referring to the state vector $\Tilde{\chi}_i$ and to the dynamical system in (\ref{eq:closed_loop_platoon_short_0}) and (\ref{eq:closed_loop_platoon_short_i}), we consider $\mu_{\Delta p}$ and $\sigma^2_{\Delta p}$ computed with respect to $\Delta\Tilde{p} = \Delta p + \Delta\Bar{p}$. For ease in notation, we introduce $e_{\Delta p} = \Delta\Bar{p}$ and $e_{\Delta v} = 0$.

\begin{lemma}\label{lmm:fmacro_bounds}
    Given the macroscopic functions $\psi^i_{\Delta p}, \ \psi^i_{\Delta v}$ in (\ref{eq:fmacro_pv}) that verify the following property:
    \begin{align}\label{eq:fmacro_pv_max_bound}
        |\psi^i_{l}(\mu_{l,i},\sigma_{l,i}^2)| &\leq c_{l} \max_{j=0,...,i}|l_j+e_l|, \: l\in\{\Delta p , \Delta v\},
    \end{align}
    then
    \begin{equation}\label{eq:psi_max_bound}
        a\psi^i_{\Delta p} + b\psi^i_{\Delta v} \leq c_{\chi}\max_{j=0,...,i}|\Tilde{\chi}_j|
    \end{equation}
    for some constants $c_{\Delta p},c_{\Delta v},c_{\chi}>0$. Moreover, if functions $\psi^i_{\Delta p}, \ \psi^i_{\Delta v}$ are bounded by
    \begin{align}\label{eq:fmacro_pv_sum_bound}
        |\psi^i_{l}(\mu_{l,i},\sigma_{l,i}^2)| &\leq \sum_{j=0}^i k^l_{ij}|l_j+e_l|, \: l\in\{\Delta p , \Delta v\},
    \end{align}
    then
    \begin{equation}\label{eq:psi_sum_bound}
        a\psi^i_{\Delta p} + b\psi^i_{\Delta v} \leq \sum_{j=0}^i k^{\chi}_{ij}|\Tilde{\chi}_j|.
    \end{equation}
    for some constants $k^{\Delta p}_{ij},k^{\Delta v}_{ij},k^{\chi}_{ij}>0$.
\end{lemma}
\begin{proof}
    See Appendix \ref{app_proofLemma1}.
\end{proof}

In order to ensure platoon string stability, we propose the macroscopic functions
\begin{equation}\label{eq:psi_delta_p}
    \psi_{\Delta p}^i = \gamma_{\Delta p} sign(\Delta\Bar{p}+\mu_{\Delta p,i})\sqrt{\sigma^2_{\Delta p,i}},
\end{equation}
\begin{equation}\label{eq:psi_delta_v}
    \psi_{\Delta v}^i = \gamma_{\Delta v} sign(\mu_{\Delta v,i})\sqrt{\sigma^2_{\Delta v,i}}
\end{equation}
where $\gamma_{\Delta p},\gamma_{\Delta v} > 0$ are constant parameters, and $\mu_{\Delta p,i}$, $\mu_{\Delta v,i}$, $\sigma^2_{\Delta p,i}$ and $\sigma^2_{\Delta v,i}$ are defined in (\ref{eq:meanvar_distance}) and (\ref{eq:meanvar_speederror}). The functions $\psi^i_{\Delta p}$ and $\psi^i_{\Delta v}$ are such that:
\begin{align}\label{eq:fmacro_pv_meanvar_property}
    \psi_{l}^i(\mu_{l,i},0) &= 0 , \ \ \psi_{l}^i(e_l,\sigma_{l,i}) = 0, \: l\in\{\Delta p , \Delta v\},
\end{align}
and $e_l$ introduced before.

Given the chosen macroscopic functions in (\ref{eq:psi_delta_p}) and (\ref{eq:psi_delta_v}), we can give the following result:
\begin{lemma}\label{lmm:psi_bounds}
    Macroscopic functions $\psi_{\Delta p}^i$ and $\psi_{\Delta v}^i$ in (\ref{eq:psi_delta_p}) and (\ref{eq:psi_delta_v}) are bounded by:
    \begin{equation}\label{eq:psi_pv_properties}
        |\psi_{l}^i| \leq \gamma_{l}\max_{j=0,...,i}|l_j+e_l|;  \:\:
        |\psi_{l}^i| \leq \gamma_{l}\sum_{j=0}^i |l_j+e_l|, \:\: l\in\{\Delta p, \Delta v\}.
    \end{equation}
\end{lemma}
\begin{proof}
    See Appendix \ref{app_proofLemma2}.
\end{proof}

We remark that the considered approach is general with respect to the choice of the macroscopic functions $\psi_{\Delta p}^i$ and $\psi_{\Delta v}^i$. Indeed, macroscopic functions different from the ones in \eqref{eq:psi_delta_p} and \eqref{eq:psi_delta_v} can be used, e.g. as done in \cite{Iovine2015NAHS}, provided that conditions in \eqref{eq:psi_pv_properties} are satisfied.

\subsection{Control strategy for constant spacing policy}

In order to embed macroscopic information in the control law $u^{cp}_i$ associated to the $i$-th vehicle, $\forall \ i\in\mathcal{I}_N^0$, and implementing the constant spacing policy in (\ref{eq:constant_spacing_policy}), we explicit the dynamic system in (\ref{eq:rho_dynamic_system}) as:
\begin{equation}\label{eq:rho_dynamic_system_CP}\color{black}
    \begin{cases}
        \Dot{\rho}_i = -\lambda\rho_i + a\psi^{i-1}_{\Delta p} + b\psi^{i-1}_{\Delta v} \\
        \rho(0) = 0
    \end{cases}
\end{equation}
with system dimension $r=1$, state transition matrix $\Lambda = -\lambda$ and input matrix $G_{\rho} = [ \: a \: b \: ]$; where $a,b \geq 0$ are chosen parameters. For the constant spacing policy we define $\rho^M_i = \rho_i$. The control law $u^{cp}_i$ is

\begin{equation}\label{eq:control_input_CP}
    u^{cp}_i = u_{i-1} + \Delta\dot{v}_i^{cp,r} - K^{cp}_{\Delta v}(\Delta v_i - \Delta v_i^{cp,r}) - (\Delta p_i - \Delta p_i^r)-\rho_{i}
\end{equation}
with equal constant gains $K^{cp}_{\Delta p},K^{cp}_{\Delta v}>0$, for each $i\in\mathcal{I}_N^0$, $\Delta p_i^r$ defined as in (\ref{eq:constant_spacing_policy}), and
\begin{align}
    \Delta v_i^{cp,r} &= -K^{cp}_{\Delta p}(\Delta p_i - \Delta p_i^r), \label{eq:delta_v_ref_CP} \\
    \Delta\dot{v}_i^{cp,r} &= -K^{cp}_{\Delta p}\Delta v_i. \label{eq:dot_delta_v_ref_CP}
\end{align}
The closed loop dynamics with respect to the extended state vector $\hat{\chi}_i$ in (\ref{eq:chi_extended_and_equilibrium}) is:

\begin{align} \label{eq:car_following_dynamics_i_closed_loop_CP}
    \dot{\hat{\chi}}_i  = \left[  
        \begin{array}{c}
            \Delta\dot{p}_i  \\
            \Delta\dot{v}_i  \\
            \dot{\rho}_{i}
        \end{array}
    \right]
    = \left[  
        \begin{array}{c}
            \Delta v_i  \\
            (*) \\
            -\lambda\rho_{i} + a\psi_{\Delta p}^{i-1} + b\psi_{\Delta v}^{i-1}
        \end{array}
    \right]
\end{align}
with 
\begin{align*}
    (*) = \Delta\dot{v}_i^{cp,r} - K^{cp}_{\Delta v}(\Delta v_i - \Delta v_i^{cp,r}) - (\Delta p_i + \Delta p_i^r) - \rho_{i}.
\end{align*}
Referring to the description of $\Tilde{\chi}_i=\hat{\chi}_i-\hat{\chi}_{e,i}$ as in (\ref{eq:closed_loop_platoon_short_0}) and (\ref{eq:closed_loop_platoon_short_i}), we remark that
\begin{equation}\label{eq:car_following_dynamics_g_cl}
    g_{cl,i}(\Tilde{\chi}_{i-1},\Tilde{\chi}_{i-2},...,\Tilde{\chi}_0) = \left[ 
        \begin{array}{c}
            0  \\
            0  \\
            a\psi_{\Delta p}^{i-1} + b\psi_{\Delta v}^{i-1}
        \end{array}
    \right].
\end{equation}

\begin{theorem}\label{thm:StringStability_CP}
    Consider the closed loop system in 
    (\ref{eq:car_following_dynamics_i_closed_loop_CP}). Given the parameters $K^{cp}_{\Delta p}>0,\ K^{cp}_{\Delta v}>0,\ \lambda>0$, then the origin of each isolated subsystem is exponentially stable. Moreover, there exist functions $\beta^{cp}\in\mathcal{KL}$  and $\gamma^{cp}\in\mathcal{K}_\infty$ such that
    \begin{equation}\label{eq:ISS_platoon_property_CP}
        |\Tilde{\chi}_i(t)| \leq \beta^{cp}(|\Tilde{\chi}_i(0)|,t) + \gamma^{cp}\left( \max_{j=0,...,i-1}|\Tilde{\chi}_j(\cdot)|^{[0,t]}_{\infty} \right)
    \end{equation}
    $\forall \ t \geq 0$ and $\gamma^{cp}(s) = \Tilde{\gamma}^{cp}s, \ s \geq 0, \ \Tilde{\gamma}^{cp}\in\mathbb{R}^+$. Also, there exist $a$ and $b$ in (\ref{eq:rho_dynamic_system_CP}) such that $\Tilde{\gamma}^{cp}\in (0,1)$. Consequently, the closed loop system in 
    (\ref{eq:car_following_dynamics_i_closed_loop_CP}) is Asymptotically String Stable.
\end{theorem}
\begin{proof}
    See Appendix \ref{app:theorem_CP}.  
\end{proof}
The mesoscopic control law in (\ref{eq:control_input_CP}) adopting a constant spacing policy is shown to ensure Asymptotic String Stability with respect to the closed loop system describing the vehicular platoon.

\subsection{Control strategy for variable spacing policy}

In this more general case, in order to obtain String Stability, we choose for  (\ref{eq:rho_dynamic_system}) 
an asymptotically stable dynamics with $r = 2$:

\begin{equation}\label{eq:rho_dynamic_system_VP}\color{black}
    \begin{cases}
        \Dot{\rho}_{1,i} = -\lambda_1 \rho_{1,i} + \rho_{2,i} \\
        \Dot{\rho}_{2,i} = -\lambda_2 \rho_{2,i} + a\psi^{i-1}_{\Delta p} + b\psi^{i-1}_{\Delta v} \\
        %
        %
        \rho_{1,i}(0) = \rho_{2,i}(0) = 0
    \end{cases}
\end{equation}
where
\begin{equation*}
    \Lambda = \left[ 
        \begin{array}{cc}
            -\lambda_1 & 1 \\
            0 & -\lambda_2
        \end{array}
    \right], \quad
    G_\rho = \left[  
        \begin{array}{cc}
            0 & 0  \\
            a & b
        \end{array}
    \right]
\end{equation*}
with $a,b \geq 0$ chosen parameters, and $\lambda_1, \ \lambda_2 > 0$.

The control law $u^{vp}_i$ associated to the $i$-th vehicle, $\forall \ i\in\mathcal{I}_N^0$, implementing the variable spacing policy in (\ref{eq:mesoscopic_spacing_policy}) is:
\begin{align}\label{eq:control_input_VP}
    \nonumber u_i &= u_{i-1} - (\Delta p_i - \Delta p_i^r) - K^{vp}_{\Delta v}(\Delta v_i - \Delta v_i^{vp,r}) \\
    \nonumber &\quad + (K^{vp}_{\Delta p} - \lambda_1)(\lambda_1\rho_{1,i} - \rho_{2,i}) +\lambda_2\rho_{2,i} \\
    &\quad -K^{vp}_{\Delta p}\Delta v_i - a\psi^{i-1}_{\Delta p} - b\psi^{i-1}_{\Delta v},
\end{align}
with equal constant gains $K^{vp}_{\Delta p},K^{vp}_{\Delta v}>0$  for each $i\in\mathcal{I}_N^0$, $\Delta p_i^r$ defined as in (\ref{eq:mesoscopic_spacing_policy}), and
\begin{align}
    \Delta v_i^{vp,r} &= \lambda_1\rho_{1,i}-\rho_{2,i}-K^{vp}_{\Delta p}(\Delta p_i-\Delta p_i^r)\label{eq:delta_v_ref_VP}.
\end{align}
To analyze the String Stability of the closed loop system, we consider the extended leader-follower state vector $\hat{\chi}_i$ and its corresponding equilibrium point $\hat{\chi}_{e,i}$ in (\ref{eq:chi_extended_and_equilibrium}). 
For each vehicle $i\in\mathcal{I}_N^0$, the resulting closed loop dynamics are:
\begin{align}\label{eq:car_following_dynamics_i_closed_loop_VP}
    \dot{\hat{\chi}}_i = \left[ 
        \begin{array}{c}
            \Delta\dot{p}_0  \\
            \Delta\dot{v}_0  \\
            \dot{\rho}_{1,i} \\
            \dot{\rho}_{2,i}
        \end{array}
    \right] = \left[  
        \begin{array}{c}
            \Delta v_i  \\
            (*) \\
            -\lambda_1\rho_{1,i} + \rho_{2,i} \\
            -\lambda_2\rho_{2,i} + a\psi_{\Delta p}^{i-1} + b\psi_{\Delta v}^{i-1}
        \end{array}
    \right]
\end{align}
with 
\begin{align*}
    (*) &= - (\Delta p_i - \Delta p_i^r) - K^{vp}_{\Delta v}(\Delta v_i - \Delta v_i^{vp,r}) \\
    \nonumber &\quad + (K^{vp}_{\Delta p} - \lambda_1)(\lambda_1\rho_{1,i} - \rho_{2,i}) +\lambda_2\rho_{2,i} \\
    &\quad -K^{vp}_{\Delta p}\Delta v_i - a\psi^{i-1}_{\Delta p} - b\psi^{i-1}_{\Delta v},
\end{align*}
Similarly to the case of a constant spacing policy, since $g_{cl,0}(\Tilde{\chi}_{-1}) = 0$, we can rewrite the system in (\ref{eq:car_following_dynamics_i_closed_loop_VP}) as (\ref{eq:closed_loop_platoon_short_0}) and (\ref{eq:closed_loop_platoon_short_i}), where 
\begin{equation}\color{black}
    g_{cl,i}(\Tilde{\chi}_{i-1},...,\Tilde{\chi}_0) = \left[
        \begin{array}{c}
            0  \\
            -\psi^{i-1}_{\Delta p} - \psi^{i-1}_{\Delta v} \\
            0 \\
            \psi^{i-1}_{\Delta p} + \psi^{i-1}_{\Delta v}
        \end{array}
    \right].   
\end{equation}
The following result holds:
\begin{theorem}\label{thm:StringStability_VP}
    Consider the closed loop system described by 
    (\ref{eq:car_following_dynamics_i_closed_loop_VP}). Given the parameters $K^{vp}_{\Delta p} > 0, \ K^{vp}_{\Delta v} > 0, \ \lambda_1,\lambda_2 > 0$, then the origin of each isolated subsystem is exponentially stable. Moreover, there exist functions $\beta^{vp}$ of class $\mathcal{KL}$ and $\gamma^{vp}$ of class $\mathcal{K}_\infty$ such that
    \begin{equation}\label{eq:ISS_platoon_property_VP}
        |\Tilde{\chi}_i(t)| \leq \beta^{vp}(|\Tilde{\chi}_i(0)|,t) + \gamma^{vp}\left( \max_{j=0,...,i} |\Tilde{\chi}_j(\cdot)|^{[0,t]}_\infty \right)
    \end{equation}
    $\forall t \geq 0$ and $\gamma^{vp}(s) = \Tilde{\gamma}^{vp}s, \ s \geq 0, \ \Tilde{\gamma}^{vp}\in\mathbb{R}^+$. Also, there exist $a$ and $b$ in (\ref{eq:rho_dynamic_system_VP}) such that $\Tilde{\gamma}^{vp}\in(0,1)$. Consequently, the closed loop system in 
    (\ref{eq:car_following_dynamics_i_closed_loop_VP}) is Asymptotically String Stable. 
\end{theorem}
\begin{proof}

    See Appendix \ref{app:theorem_VP}.
    
\end{proof}

The mesoscopic controller in (\ref{eq:control_input_VP}) adopting a time-varying spacing policy is shown to ensure Asymptotic String Stability with respect to the closed loop system describing the vehicular platoon.

We applied the general framework described in Theorem  \ref{thm:StringStability} to two different spacing policies. Rather than the philosophical choice of the desired spacing policy, we remark that the differences on the proposed control action in (\ref{eq:control_input_CP})  and (\ref{eq:control_input_VP}) mainly rely on the possibility to have the functions in (\ref{eq:fmacro_pv}) directly affecting the control input. A proper choice of the parameters $a$ and $b$ in (\ref{eq:rho_dynamic_system_CP}) and (\ref{eq:rho_dynamic_system_VP}) will size the real-time contribution of the macroscopic functions (\ref{eq:psi_delta_p}) and (\ref{eq:psi_delta_v}) with respect to the filtered version given by $\rho_i$ in (\ref{eq:rho_dynamic_system}). The differences among the two approaches and the resulting control laws are better highlighted in the next section where some simulations are illustrated.

\section{Simulations}\label{sct:simulations}

The two introduced control strategies are simulated in Matlab\&Simulink. Based on the modeling in (\ref{eq:car_following_dynamics_i}), we consider a platoon of $N+1=31$ vehicles. The initial conditions for each vehicle are randomly generated in a neighborhood of the equilibrium point. It results $\mu_{\Delta p}\neq-\Delta\Bar{p},\mu_{\Delta v}\neq 0$ $\sigma^2_{\Delta p},\sigma^2_{\Delta v} \neq 0$. The constant reference distance is $\Delta\Bar{p} = 20$ m and the initial desired speed of the leading vehicle is $\Bar{v} = 14$ m/s. The vehicle speed is such that $0\leq v_i\leq 36$ m/s, and the acceleration is bounded such that $-4\leq u_i\leq 4$ m/s\textsuperscript{2}. The control parameters are chosen to obtain a value $\Tilde{\gamma}\approx 0.5$ for both controllers. To better stress the proposed controllers, we analyze the behavior of the system when a disturbance acts on the acceleration of vehicle $i=0$, and it is not communicated to vehicle $i=1$. Both controllers are simulated with the same simulation time of 1 minute, initial condition and perturbed leader vehicle. We split the simulation time in four phases:
\begin{enumerate}

    \item From $t=t_0=0$ s to $t=t_1=10$ s: the vehicles starts with initial conditions that are different from the desired speed and the desired distance. No disturbance is acting on the leader vehicle, and its desired speed is the initial one, i.e. $\Bar{v} = 14$ m/s.
    
    \item From $t=t_1=10$ s to $t=t_2=25$ s: the leader tracks a variable speed reference. From $t=10$ s to $t=20$ s the desired speed is $\Bar{v}=25$ m/s, while from $t=20$ s to $t=25$ s it is $\Bar{v}=20$ m/s.
    
    \item From $t=t_2=25$ s to $t=t_3=30$ s: a disturbance acts to the acceleration of the first vehicle $i=0$. At $t_2$ a positive pulse of amplitude $4$ m/s\textsuperscript{2} and length $5$ s is considered. The vehicle succeeds to properly counteract to it, but the control input of $i=0$ being saturated there is not the possibility to return to the desired speed. Since the disturbance is an external input, it is not communicated to the follower and could propagate along the platoon.
    
    \item From $t=t_3=35$ s to $t=t_4=60$ s: the leader tracks a variable speed reference while being subject to a sinusoidal disturbance acting on its acceleration. The disturbance has an amplitude of $2$ m/s and a frequency of $1$ rad/s. From $t=35$ s to $t=45$ s the desired speed is $\Bar{v}=14$ m/s, while from $t=45$ s to $t=60$ s it is $\Bar{v}=25$ m/s. 
    
\end{enumerate}

\noindent In the following, we use a color scale from light yellow to dark red to represent the set of vehicles in the figures, ranging from $i=0$ (light yellow) to $i=N$ (dark red).


\begin{figure}[h]
    \centering
    \includegraphics[width = 1\columnwidth]{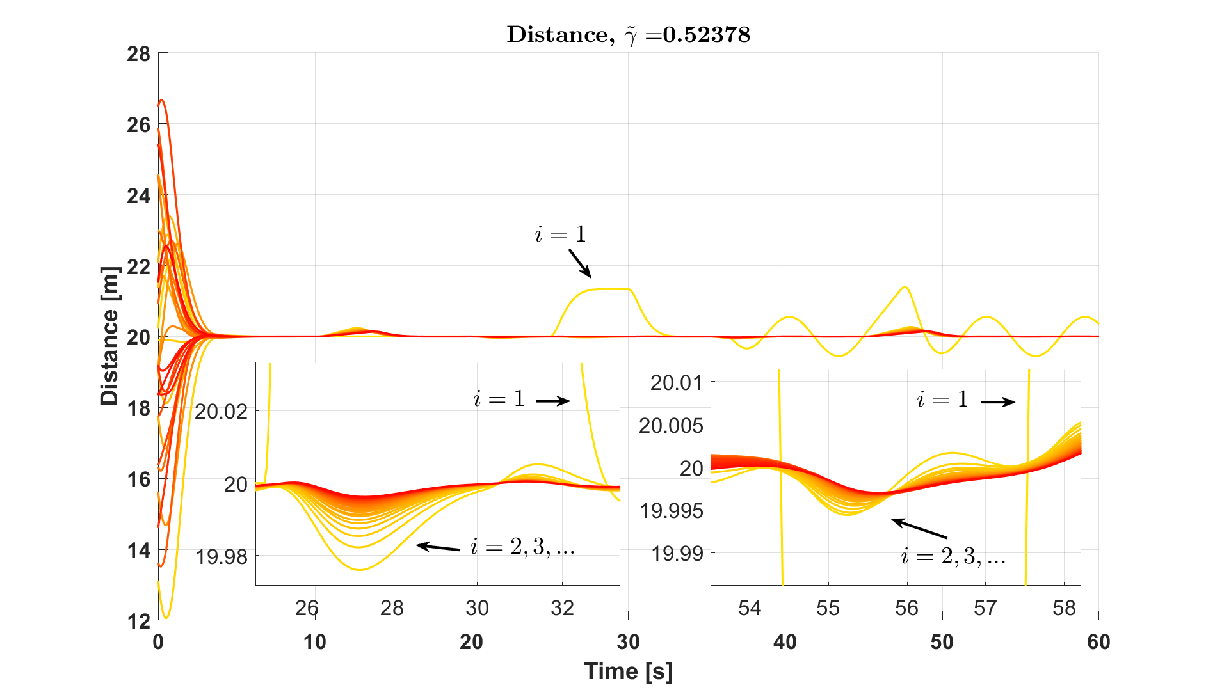}
    \caption{Control strategy for constant spacing policy: Distances. The color scale from light yellow to dark red represents the inter-vehicular distances between the vehicles of the platoon from the head pair $(0,1)$ to the tail one $(N-1,N)$.}
    \label{fig:distances_CP}
    \centering
    \includegraphics[width = 1\columnwidth]{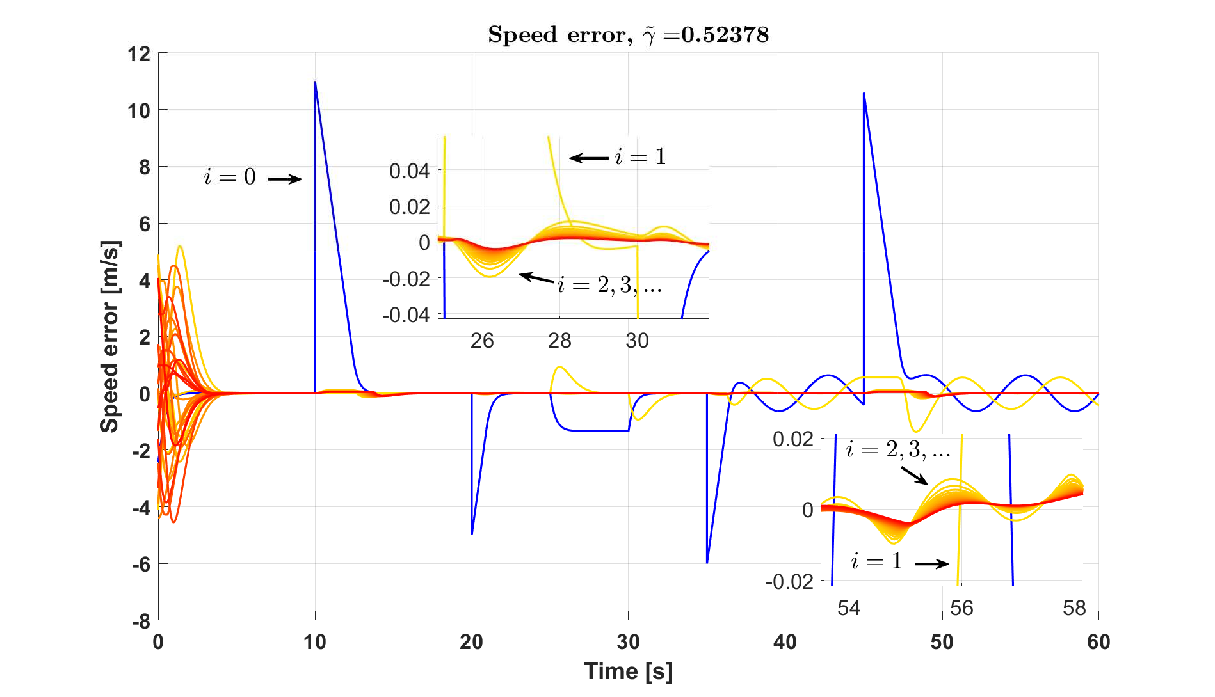}
    \caption{Control strategy for constant spacing policy: Speed Errors. The color scale from light yellow to dark red represents the inter-vehicular speed differences between the vehicles of the platoon from the head pair $(0,1)$ to the tail one $(N-1,N)$. The blue line represents the speed difference between vehicle $i=0$ and the reference speed $\Bar{v}$.}
    \label{fig:speed_differences_CP}
    \centering
    \includegraphics[width = 1\columnwidth]{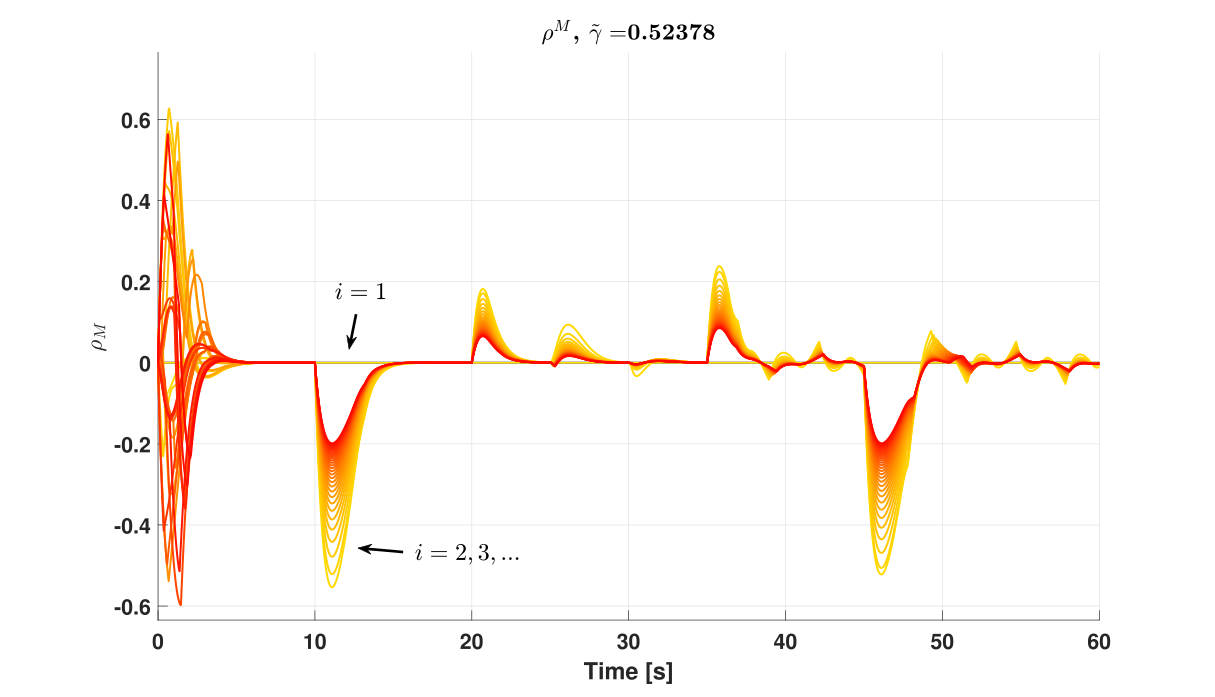}
    \caption{Control strategy for constant spacing policy: $\rho_i^M$. The color scale from light yellow to dark red represents the variable $\rho_i^M$ associated to the $i$-th vehicles of the platoon from $i=0$ to $i=N$.}
    \label{fig:rho_M_CP}
\end{figure}

\subsection{Control strategy for constant spacing policy}

Figures \ref{fig:distances_CP}, \ref{fig:speed_differences_CP} and \ref{fig:rho_M_CP} show, respectively, the inter-vehicular distances ($p_{i-1}-p_i = -\Delta p_i$), inter-vehicular speed differences ($v_{i-1}-v_i = -\Delta v_i$) and the macroscopic variables ($\rho_i^M$) profiles for each vehicle of the platoon when the control input (\ref{eq:control_input_CP}) is implemented. 

In the first phase, the vehicles are shown to converge to the desired distance and speed, corresponding to zero inter-vehicular speed errors, with a fast transient. 

In the second phase, the desired speed profile has high steps resulting in high peaks at $10$ s and $20$ s for $\Bar{v}-v_0 = -\Delta v_0$ (see Figure \ref{fig:speed_differences_CP}, blue line). However, the vehicles succeed to track the variable speed profile, as shown in Figures \ref{fig:distances_CP} and \ref{fig:speed_differences_CP}. Figure \ref{fig:rho_M_CP} depicts how the macroscopic function $\rho_i^M$ catches the variations due to the reference change for the leading vehicle. We remark that the dark red line, corresponding to vehicle $N$, is the most attenuated one due to the platoon damping action. Indeed, the followers filter their action due to the car-following interaction with respect to the whole platoon behavior. The consideration of such macroscopic information causes the vehicles to lightly adapt their inter-vehicular distance (see Figure \ref{fig:distances_CP}) and speed difference reference (see Figure \ref{fig:speed_differences_CP}), in order to prevent the perturbation given by the reference speed variation to be amplified.

In the third phase, vehicle $i=0$ is not able to counteract the acting pulse disturbance and is therefore forced to accelerate (see Figure \ref{fig:speed_differences_CP}, blue line). In this phase, the controller of $i=1$ does not know the correct value of $u_0$. However it converges to the same speed of $i=0$ after a small transient of three seconds (see Figure  \ref{fig:speed_differences_CP}). Since no macroscopic information is available to it, it does not succeed to perfectly track the desired distance when the positive disturbance acts on the leader acceleration in $[t_2,t_3)$ (see Figure \ref{fig:distances_CP}). Finally, at $t=30$ s the disturbance is not active anymore and the leader can restore its desired speed. Also, $i=1$ receives correct information about its leader acceleration and is able to return to the ideal distance. The dynamical evolution of distances and speed differences of the remaining vehicles in the platoon during the generated transients after $t_2$ and $t_3$ catches the contribution of the macroscopic information. Indeed, it is possible to remark an anticipatory behavior since the vehicles along the platoon balance the propagating error by anticipating their action.  In Figure \ref{fig:rho_M_CP}, it is possible to see how the evolution of $\rho_i^M$ reacts when the disturbance take places, and its attenuation.

In the fourth phase, vehicle $i=0$ successfully tracks its time-varying reference, although being subject to a sinusoidal disturbance that does not allow to reach a constant speed in steady state (see Figure \ref{fig:speed_differences_CP}, blue line). As for the pulse disturbance in the third phase, vehicle $i=1$ is not able to properly counteract to vehicle $i=0$ because of the incorrect received value for $u_0$. Consequently, the distance profile for the first pair of vehicles is not the desired one, as shown in Figure \ref{fig:distances_CP}. However, the increase  of considered information along the platoon let the next following vehicles have a better response with respect to the first follower. Figure \ref{fig:distances_CP} clearly shows that the remaining vehicles composing the platoon have a smoother distance profile. They are capable to filter the sinusoidal disturbance thanks to the utilisation of macroscopic information. Indeed, consistent variation from the desired distance profile take place only when both the sinusoidal disturbance and the step one affect together the first vehicle reference. The utilisation of the macroscopic information results in a precious filter that creates an anticipatory behavior with respect to undesired changes taking place along the platoon. We remark that, as expected,  the last vehicle is the one that is less affected by the sinusoidal disturbance. This is due to the propagation of the  contribution of the macroscopic information, which decreases along the platoon as shown in Figure \ref{fig:rho_M_CP}.

\subsection{Control strategy for variable spacing policy}

\begin{figure}[h]
    \centering
    \includegraphics[width = 1\columnwidth]{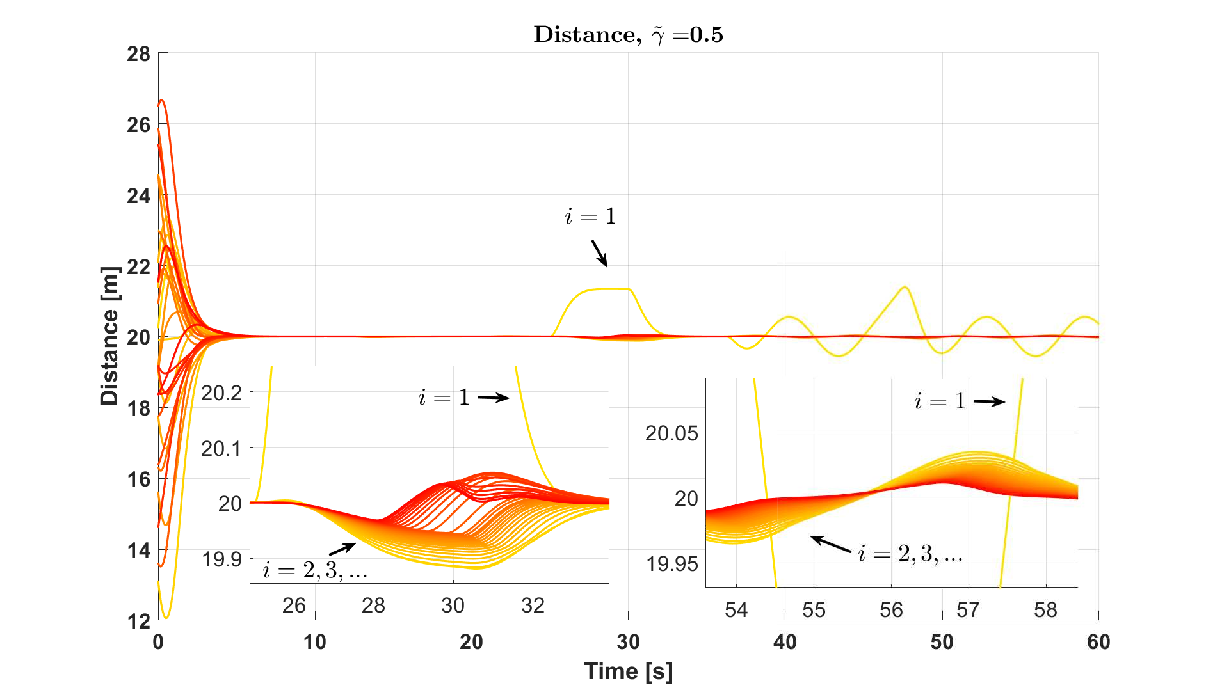}
    \caption{Control strategy for variable spacing policy: Distances. The color scale from light yellow to dark red represents the inter-vehicular distances between the vehicles of the platoon from the head pair $(0,1)$ to the tail one $(N-1,N)$.}
    \label{fig:distances_VP}
    \centering
    \includegraphics[width = 1\columnwidth]{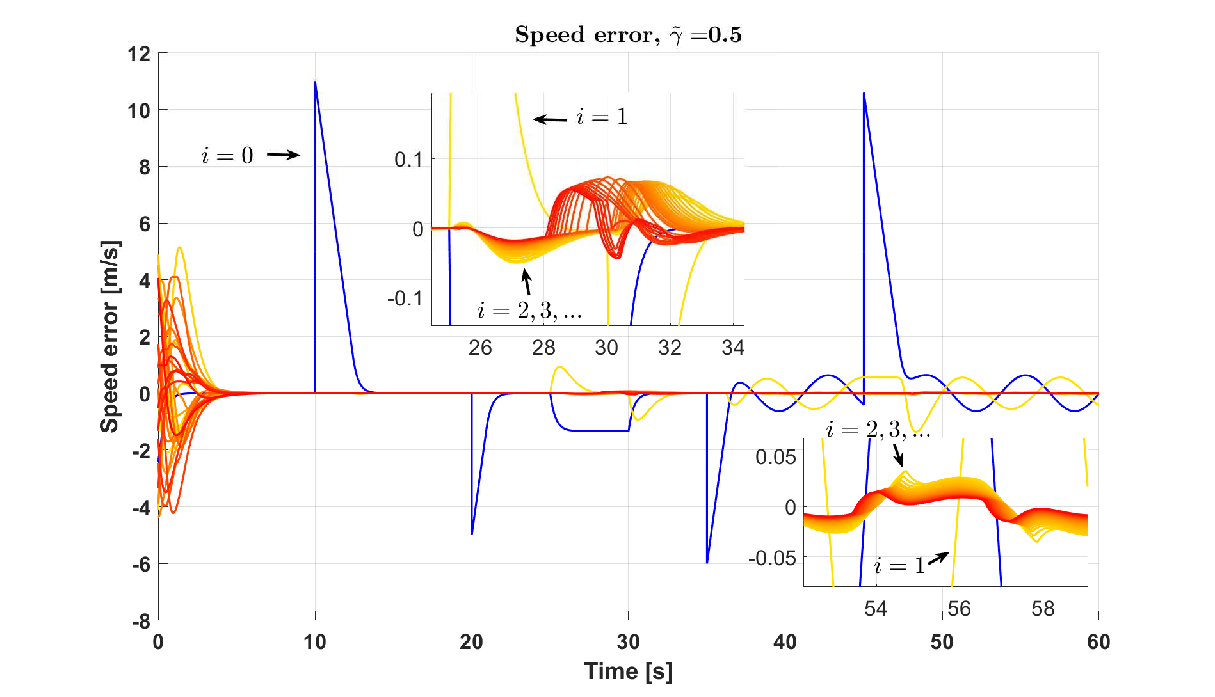}
    \caption{Control strategy for variable spacing policy: Speed Errors. The color scale from light yellow to dark red represents the inter-vehicular speed differences between the vehicles of the platoon from the head pair $(0,1)$ to the tail one $(N-1,N)$. The blue line represents the speed difference between vehicle $i=0$ and the reference speed $\Bar{v}$.}
    \label{fig:speed_differences_VP}
    \centering
    \includegraphics[width = 1\columnwidth]{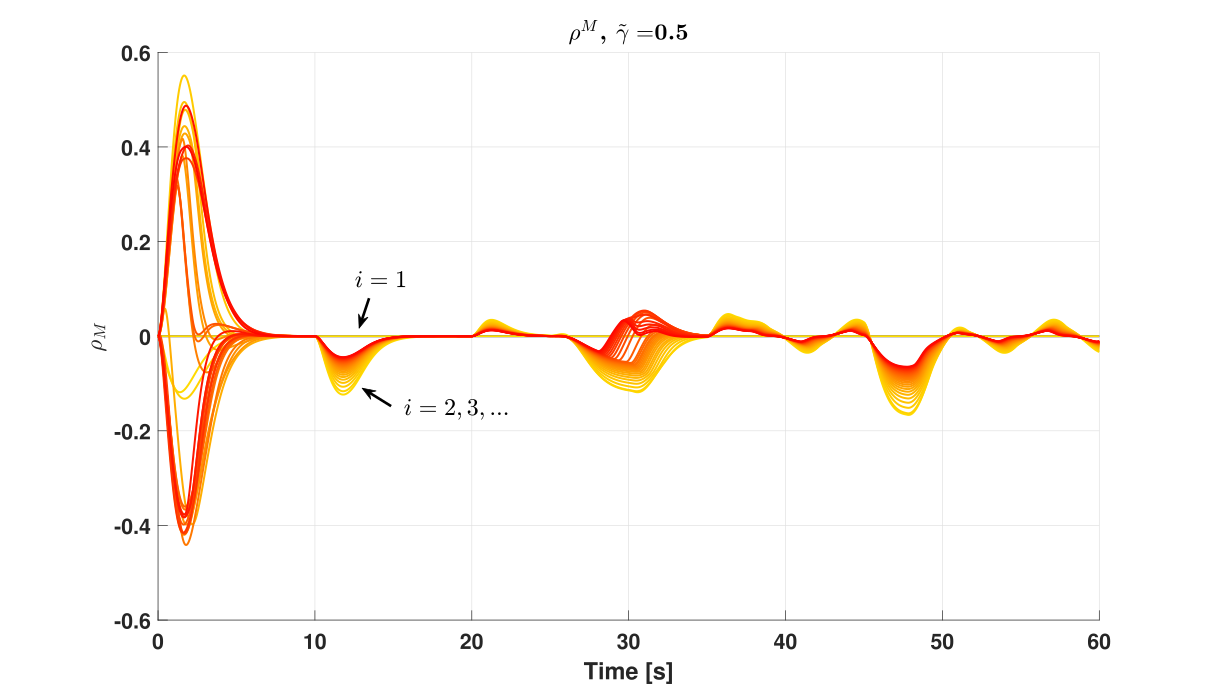}
    \caption{Control strategy for variable spacing policy: $\rho_i^M$. The color scale from light yellow to dark red represents the variable $\rho_i^M$ associated to the $i$-th vehicles of the platoon from $i=0$ to $i=N$.}
    \label{fig:rho_M_VP}
\end{figure}

Figures \ref{fig:distances_VP}, \ref{fig:speed_differences_VP} and \ref{fig:rho_M_VP} show the platoon dynamical evolution when the control input (\ref{eq:control_input_VP}) is implemented. The simulation time is split in the same 4 phases of the constant spacing policy case. 

In the first phase, the vehicles are shown to converge to the desired speed and distance with a fast transient, similarly to the constant spacing policy case.

Then, during the second phase, the vehicles succeed to track the variable speed profile and to remain at the desired distance. Since the macroscopic function is directly used to adapt the inter-vehicular distance, we see in Figure \ref{fig:rho_M_VP}  at $10$ s and $20$ s a smaller overshoot with respect to Figure \ref{fig:rho_M_CP} describing the constant spacing policy case. Consequently, Figures \ref{fig:distances_VP} and \ref{fig:speed_differences_VP} show a better performance for the distance and speed difference profile with respect to such overshoot. The perturbations due to the reference speed variation for the first vehicle are better attenuated in the variable spacing policy case.

In the third phase, it is then possible to better appreciate the effect of using the macroscopic variable to vary the desired distance. Indeed, when the disturbance acts on $i=0$ causing $i=1$  not to perfectly track the desired distance, the remaining set of followers adapt their inter-vehicular distances such that the perturbation is reduced through the string (see Figure \ref{fig:distances_VP}). In the variable spacing policy case, the distance adaptation is higher with respect to the constant one. Figure \ref{fig:speed_differences_VP} shows that the vehicles along the speed difference error scales in intensity along the platoon, resulting in an anticipatory action. Furthermore, when the disturbance disappear at $30$ s, the platoon rapidly converges to the correct reference speed. Since we exploit the macroscopic variable to vary the desired distance, we remark a similar evolution in Figure \ref{fig:distances_VP} with respect to Figure \ref{fig:rho_M_VP} in the interval $[t_2,t_3)$. It means that the macroscopic information described by $\rho_i^M$ is correctly used to stabilize the platoon.

In the fourth phase, an even better dynamical behavior is produced with respect to the performing one of the constant spacing policy. Indeed, the filtering of the sinusoidal disturbance is more effective for both distance and speed difference references, as shown in Figures \ref{fig:distances_VP} and \ref{fig:speed_differences_VP}). Indeed, expect for the vehicle $i=1$ that cannot correctly track the desired distance due to the lack of macroscopic information, the rest of the platoon is able to attenuate the propagating oscillations and to counteract to them in a more performing way, as shown in the zooms of Figures \ref{fig:distances_VP} and \ref{fig:speed_differences_VP}. Then, also in this phase the vehicles in the platoon correctly anticipate their control action to neutralise the disturbance deriving from the head vehicle.

\subsection{Comments on simulations}

The proposed control laws in (\ref{eq:control_input_CP}) and (\ref{eq:control_input_VP}) manage to safely control a platoon of vehicles and exploit the information resulting from the macroscopic variable. The control inputs provide transient harmonization on the whole traffic throughput while ensuring String Stability. The dynamical evolution shows a reduction of the oscillations propagation along the platoon, in the nominal case as well as in the presence of an active external disturbance, in both cases of constant and variable spacing policies. The utilization of variables that aggregate macroscopic information proves to be a powerful tool that acts as a smoothing filter with respect to disturbances in the distance and speed error references. In particular, it empowers the constant spacing policy with the String Stability property without the knowledge of specific data, such as the information about the platoon leader vehicle. Indeed, for a constant spacing policy that does not make use of any aggregate information, it is necessary to share both platoon leader velocity and acceleration to ensure String Stability, see e.g. \cite{Hedrick2004strings}. Moreover, the time-varying spacing policy depending on macroscopic information may increase the traffic throughput, on the contrary to classical results as in \cite{Ploeg2014TCST}. Indeed, it only modifies  the transients, and consequently allows for a reduction in the steady-state inter-vehicular distance with no dependence on the velocity profile. 

Future work will focus on reducing the amount of information exchanged to propagate macroscopic quantities, by investigating V2I frameworks with aggregate and recursive variables. Moreover, the combination of estimation techniques for the macroscopic quantities with the developed control framework will be investigated.

\vspace{-50pt}
\begin{center}
\begin{table}
    \centering
    \caption{The control parameters: constant spacing policy.}
    \begin{tabular}{ | c | c | c | c |c | c | }
        \hline
        Parameter  & Value & Parameter & Value & Parameter & Value  \\ \hline
        $K_{\Delta p}$ & $1$ & $K_{\Delta v}$ & $2$  & $\Upsilon$ & $0.9$   \\ \hline
        $\lambda$ & $1.5$ & $a$ & $0.5$ & $b$ & $0.5$  \\  \hline
        $\gamma_{\Delta p}$ & $0.5$ & $\gamma_{\Delta v}$ & $0.5$ & $\Tilde{\gamma}$ & $0.52$\\  
        \hline
    \end{tabular}\label{table:parameters_CVP}
\end{table}
\end{center}

\begin{center}
\begin{table}
\centering
    \caption{The control parameters: variable spacing policy.}
    \begin{tabular}{ | c | c | c | c |c | c | }
    \hline
  Parameter  & Value & Parameter & Value & Parameter & Value  \\ \hline
  $K_{\Delta p}$ & $1$ & $K_{\Delta v}$ & $2$  & $\Upsilon$ & $0.9$   \\ \hline
  $\lambda_1,\lambda_2$ & $1.5$ & $a$ & $1$ & $b$ & $0.2$  \\  \hline
  $\gamma_{\Delta p}$ & $0.5$ & $\gamma_{\Delta v}$ & $0.5$ & $\Tilde{\gamma}$ & $0.5$\\  \hline
    \end{tabular}\label{table:parameters_CV}
    \end{table}
\end{center}

\vspace{10pt}
\section{Conclusion}\label{sct:conclusion}

This paper describes the capability to consider macroscopic variables for improving String Stability performance of a platoon of CACC autonomous vehicles. As the variance of microscopic quantities is related to the macroscopic density, the proposed stability analysis opens the possibility to properly control a platoon by propagating macroscopic density information. Two control laws, achieving Asymptotic String Stability, with different spacing policies and based on information obtained by V2V communication have been proposed. The improvements resulting from the consideration of macroscopic information are shown by simulation results. The proposed mesoscopic control laws produce an anticipatory behaviour, which provides a better transient harmonization. 

Future work will focus on the possibility of considering only V2I communications for macroscopic quantities sharing, and on the extension of the proposed framework in a mixed traffic situation with non-communicating vehicles. Also, more complex models that include non-idealities, such as delays and model uncertainties, will be investigated.

\bibliographystyle{ieeetr}
\bibliography{biblioTraffic_20200915.bib}  

\appendix
\section{Appendices}

\subsection{Proof of Lemma \ref{lmm:fmacro_bounds}}\label{app_proofLemma1}
\begin{proof}
    Let us prove the property in (\ref{eq:psi_max_bound}). It results:
    \begin{align}
        \nonumber a\psi^i_{\Delta p}+b\psi^i_{\Delta v} &\leq a|\psi^i_{\Delta p}|+b|\psi^i_{\Delta v}| \\
        \nonumber &\leq a c_{\Delta p}\max_{j=0,...,i}|\Delta\Tilde{p}_j| + b c_{\Delta v}\max_{j=0,...,i}|\Delta v_j| \\
        \nonumber &\leq a c_{\Delta p}\max_{j=0,...,i}|\Tilde{\chi}_j| + b c_{\Delta v}\max_{j=0,...,i}|\Tilde{\chi}_j| \\
        &\leq (a c_{\Delta p} + b c_{\Delta v})\max_{j=0,...,i}|\Tilde{\chi}_j|
    \end{align}
    where $c_{\chi} = a c_{\Delta p} + b c_{\Delta v}$.
    
    The inequality in (\ref{eq:psi_sum_bound}) can be proved using similar arguments:
    \begin{align}
       \nonumber  a\psi^i_{\Delta p}+&b\psi^i_{\Delta v} \leq a|\psi^i_{\Delta p}|+b|\psi^i_{\Delta v}| \\
        \nonumber &\leq a \sum_{j=0}^i k^{\Delta p}_{ij}|\Delta\Tilde{p}_j| + b \sum_{j=0}^i k^{\Delta v}_{ij} |\Delta v_j| \\
         &\leq \sum_{j=0}^i \left| \left[
        \begin{array}{ccc}
            a k^{\Delta p}_{ij} & 0 & 0 \\
            0 & b k^{\Delta v}_{ij} & 0 \\
            0 & 0 & 0
        \end{array}
        \right] \Tilde{\chi}_j \right|_1 \leq \sum_{j=0}^i k^{\chi}_{ij}|\Tilde{\chi}_j|
    \end{align}
    where $k^{\chi}_{ij} = \sqrt{2+r}\max\{ a k^{\Delta p}_{ij} , b k^{\Delta v}_{ij} \}$.
\end{proof}

\subsection{Proof of Lemma \ref{lmm:psi_bounds}}\label{app_proofLemma2}

\begin{proof}
    Because of the similarity of conditions in (\ref{eq:psi_pv_properties}), we prove only the inequality w.r.t. $\Delta p$. First we recall the following variance property: let $l\in\{1,...,m\}$ and $y_l\in\mathbb{R}, \ |y_l| < \infty \ \forall \ l$. The variance $\sigma^2_y$ with respect to the set of values $y_l$ is such that
    \begin{equation}\label{eq:variance_property}
        \sigma_y^2 \leq \frac{1}{4}(\max_l y_l - \min_l y_l)^2.
    \end{equation}
    Then, for the macroscopic function $\psi^i_{\Delta p}$ the following inequality is proved:
    \begin{align}
        \nonumber |\psi^i_{\Delta p}| &\leq \gamma_{\Delta p}\sqrt{\sigma^2_{\Delta p,i}} \\
        \nonumber &\leq \frac{1}{2}\gamma_{\Delta p}|\max_{j=0,...,i}\Delta\Tilde{p}_i - \min_{j=0,...,i}\Delta\Tilde{p}_i| \\
        &\leq \gamma_{\Delta p}\max_{j=0,...,i}|\Delta\Tilde{p}_i|. 
    \end{align}
    It follows that $|\psi^i_{\Delta v}| \leq \gamma_{\Delta v}\max_{j=0,...,i}|\Delta v_j|$. Moreover,

    \begin{align}
        \nonumber |\psi^i_{\Delta p}| &\leq \gamma_{\Delta p}\sqrt{\sigma^2_{\Delta p,i}} \\
        \nonumber &= \gamma_{\Delta p}\left( \frac{1}{i+1}\sum_{j=0}^{i} \Delta\Tilde{p}_j^2 - \frac{1}{(i+1)^2}\left( \sum_{j=0}^{i} \Delta\Tilde{p}_j\right)^2  \right)^{\frac{1}{2}} \\
        \nonumber &\leq \gamma_{\Delta p}\frac{1}{\sqrt{i+1}}\left( \sum_{j=0}^{i} \Delta\Tilde{p}_j^2 \right)^{\frac{1}{2}} \\
        &\leq \gamma_{\Delta p}\frac{1}{\sqrt{i+1}} \sum_{j=0}^{i}|\Delta\Tilde{p}_j|
    \end{align}
    where we have used the inequality $|x|_2 \leq |x|_1$. In the same way we can prove that 
    \begin{equation}
        |\psi^{i}_{\Delta v}| \leq \gamma_{\Delta v}\frac{1}{\sqrt{i+1}} \sum_{j=0}^{i} |\Delta v_j|
    \end{equation}

\end{proof}

\subsection{Proof of Theorem \ref{thm:StringStability_CP}}\label{app:theorem_CP}

\begin{proof}
    In the first part of the proof we show that the origin of each isolated subsystem (i.e. $g_{cl,i} = 0, \ \forall \ i$) is exponentially stable. Let us consider the candidate Lyapunov function $W^{cp}_i = W^{cp}(\Tilde{\chi}_i)$ for the $i$-th dynamical system $\Tilde{\chi}_i$, $i\in\mathcal{I}_N^0$, as:
    \begin{align}\label{eq:Wi_CP}
        \nonumber W^{cp}(\Tilde{\chi}_i) &= \frac{1}{2}(\Delta p_i + \Delta\Bar{p})^2 + \frac{1}{2}(\Delta v_i + K^{cp}_{\Delta p}(\Delta p_i + \Delta\Bar{p}))^2 \\
        &\quad +\frac{1}{2}\rho^2_{i} \\
        &= \frac{1}{2} \Tilde{\chi}_i^T \underbrace{ \left[ \begin{array}{ccc}
            1+(K^{cp}_{\Delta p})^2 & 2K^{cp}_{\Delta p} & 0 \\
            0 & 1 & 0 \\
            0 & 0 & 1
        \end{array} \right] }_{P^{cp}} \Tilde{\chi}_i
    \end{align}
    Function $W^{cp}_i$ satisfies the inequalities
    \begin{equation}\label{eq:Wi_bound_CP}
        \underline{\alpha}^{cp}|\Tilde{\chi}_i|^2 \leq W_i^{cp} \leq \Bar{\alpha}^{cp}|\Tilde{\chi}_i|^2,
    \end{equation}
    where, by defining with $\lambda_{\min}(\cdot)$ and $\lambda_{\max}(\cdot)$ the minimum and maximum eigenvalues of a matrix, we obtain
    \begin{align}\label{eq:Wi_alpha_CP}
        \underline{\alpha}^{cp} &= \frac{1}{2}\lambda_{\min}(P^{cp}) = \frac{1}{2}, \\
        \Bar{\alpha}^{cp} &= \frac{1}{2}\lambda_{\max}(P^{cp}) = \frac{1}{2}(1 + (K^{cp}_{\Delta p})^2).
    \end{align}
    The time derivative of $W^{cp}_i$ in (\ref{eq:Wi_CP}) is:
    \begin{align}\label{eq:dot_Wi_CP_1}
        \nonumber \dot{W}^{cp}_i & =  -K^{cp}_{\Delta p}( \Delta p_i + \Delta\Bar{p} )^2 - K^{cp}_{\Delta v}[ \Delta v_i + K^{cp}_{\Delta p}( \Delta p_i+\Delta\Bar{p} ) ]^2 \\  
        \nonumber & \quad -\lambda\rho_i^2 - \rho_i[ \Delta v_i+K^{cp}_{\Delta p}(\Delta p_i+\Delta\Bar{p}) ] \\
  \nonumber       & = -\Tilde{\chi}_i^T \underbrace{ \left[ \begin{array}{ccc}
            p & 2 K^{cp}_{\Delta v}K^{cp}_{\Delta p} & K^{cp}_{\Delta p} \\
            0 & K^{cp}_{\Delta v} & 1 \\
            0 & 0 & \lambda
        \end{array} \right] }_{Q^{cp}} \Tilde{\chi}_i \\
        &\leq -\alpha^{cp}|\Tilde{\chi}_i|^2 
    \end{align}
    where $p = K^{cp}_{\Delta p}(1+K^{cp}_{\Delta v}K^{cp}_{\Delta p})$ and
    \begin{align}
        \alpha^{cp} &= \lambda_{\min}(Q^{cp}) = \min\left\{ K^{cp}_{\Delta v} , K^{cp}_{\Delta p}(1+K^{cp}_{\Delta v} K^{cp}_{\Delta p}) , \lambda \right\}. \label{eq:dot_Wi_alphatilde_CP}
    \end{align}
    Since $\alpha^{cp} > 0$ for hypothesis, then $W^{cp}_i$ is a Lyapunov function for the $i$-th isolated subsystem and the exponential stability is proven (see \cite[Theorem~4.10]{B_khalil_2002}).
    
    We go on by proving the ISS property in (\ref{eq:ISS_platoon_property_CP}). As before, let us consider the function $W^{cp}_i$ for the $i$-th system and its time derivative $\dot{W}^{cp}_i$. In this case, we consider an interconnection term $g_{cl,i}\neq0$. The following inequality is satisfied:
    \begin{align}\label{eq:dot_Wi_CP_2}
        \nonumber \dot{W}^{cp}_i & =  -K^{cp}_{\Delta p}( \Delta p_i + \Delta\Bar{p} )^2 - K^{cp}_{\Delta v}( \Delta v_i \\  
        \nonumber & \quad + K^{cp}_{\Delta p}( \Delta p_i+\Delta\Bar{p} ) )^2 -\lambda\rho_i^2 \\
        \nonumber & \quad - \rho_i( \Delta v_i+K^{cp}_{\Delta p}(\Delta p_i+\Delta\Bar{p}) \\
        \nonumber & \quad  + a\psi_{\Delta p}^{i-1} + b\psi_{\Delta v}^{i-1} ) \\
        & \leq -\alpha^{cp} |\Tilde{\chi}_i|^2 + |\rho_i|( a|\psi_{\Delta p}^{i-1}| + b|\psi_{\Delta v}^{i-1}| ) .
    \end{align}
    Define 
    \begin{equation}\label{eq:constant_definition_CP}
        d = a\gamma_{\Delta p}+b\gamma_{\Delta v} > 0, \ \Upsilon^{cp}\in(0,1) .  
    \end{equation}
    By Lemmas \ref{lmm:fmacro_bounds} and \ref{lmm:psi_bounds}, it results
    \begin{equation}
        |g_{cl,i}| \leq \max_{j=0,...,i-1}|\Tilde{\chi}_j|, \ \ |g_{cl,i}| \leq \sum_{j=0}^i k_{ij}|\Tilde{\chi}_j|.
    \end{equation}
    Then,
    \begin{align}\label{eq:dot_Wi_CP_3}
        \nonumber \dot{W}_i &\leq -\alpha^{cp}|\Tilde{\chi}_i|^2 + d|\Tilde{\chi}_i|\max_{j=0,...,i-1}|\Tilde{\chi}_j| + \alpha^{cp}\Upsilon^{cp} - \alpha^{cp}\Upsilon^{cp} \\
        &\leq -(1-\Upsilon^{cp})\alpha^{cp}|\Tilde{\chi}_i|^2, \quad \forall \  |\Tilde{\chi}_i|\geq\frac{d}{\alpha^{cp}\Upsilon^{cp}}\max_{j=0,...,i-1}|\Tilde{\chi}_j|.
    \end{align}
    The condition in (\ref{eq:dot_Wi_CP_3}) satisfies the ISS requirements. Consequently,  the inequality in (\ref{eq:ISS_platoon_property_CP}) is verified according to \cite[Theorem~4.19]{B_khalil_2002}. Moreover, 
    \begin{equation}\label{eq:gammatilde_CP}
        \gamma^{cp}(s) = \Tilde{\gamma}^{cp}s\ \  \forall \  s \geq 0, \ \ \  \Tilde{\gamma}^{cp} = \sqrt{\frac{\Bar{\alpha}^{cp}}{\underline{\alpha}^{cp}}}\frac{d}{\alpha^{cp}\Upsilon^{cp}} > 0.
    \end{equation}
    Since the parameters $a,b \leq 0$ in the dynamics of $\rho_i$ in (\ref{eq:rho_dynamic_system}) can be arbitrarily selected, the constant $d$ defined in (\ref{eq:constant_definition_CP}) can be chosen such that $\Tilde{\gamma}^{cp}$ in (\ref{eq:gammatilde_CP}) belongs to $(0,1)$.
    
    On the basis of Theorem \ref{thm:StringStability}, Asymptotic String Stability of the platoon can be obtained by using an appropriately chosen function describing macroscopic information.
\end{proof}

\subsection{Proof of Theorem \ref{thm:StringStability_VP}}\label{app:theorem_VP}

\begin{proof}
    
    In the first part of the proof we focus on showing that the origin of each isolated subsystem (i.e. $g_{cl,i} = 0, \ \forall i$) is exponentially stable. Let us consider the candidate Lyapunov function $W^{vp}_i = W^{vp}(\Tilde{\chi}_i)$ for the $i$-th dynamical system $\Tilde{\chi}_i$, $i\in\mathcal{I}_N^0$, as:
    \begin{align}\label{eq:Wi_VP}
        \nonumber & W^{vp}_i = \frac{1}{2}(\Delta p_i - \Delta p_i^r)^2 + \frac{1}{2}(\Delta v_i - \Delta v_i^{vp,r})^2 + \\
        \nonumber &\quad\quad\quad  + \frac{1}{2}\rho_{1,i}^2 + \frac{1}{2}\rho_{2,i}^2 \\
        &\: = \frac{1}{2} \Tilde{\chi}_i^T \underbrace{ \left[ 
            \begin{array}{cccc}
                1+(K_{\Delta p}^{vp})^2 & p_1 & p_2 & p_1 \\
                0 & 1 & p_3 & 2 \\
                0 & 0 & 2+(\lambda_1-K_{\Delta p}^{vp})^2 & p_3 \\
                0 & 0 & 0 & 2
            \end{array} \right] }_{P^{vp}} \Tilde{\chi}_i
    \end{align}
    where $p_1 = 2K_{\Delta p}^{vp}, \ p_2 = 2(1+(K^{vp}_{\Delta p})^2-\lambda_1 K_{\Delta p}^{vp}), \ p_3 = 2(K_{\Delta p}^{vp}-\lambda_1)$. Function $W^{vp}_i$ satisfies the inequalities
    \begin{equation}\label{eq:Wi_bound_VP}
        \underline{\alpha}^{vp}|\Tilde{\chi}_i|^2 \leq W_i^{vp} \leq \Bar{\alpha}^{vp}|\Tilde{\chi}_i|^2
    \end{equation}
    \begin{align}\label{eq:Wi_alpha_VP}
        \underline{\alpha}^{vp} &= \frac{1}{2}\lambda_{\min}(P^{vp}) = \frac{1}{2}, \\
        \nonumber \Bar{\alpha}^{vp} &= \frac{1}{2}\lambda_{\max}(P^{vp}) \\
        &= \frac{1}{2}\max \Big\{ 1+(K^{vp}_{\Delta p})^2 , 2+(\lambda_1-K_{\Delta p}^{vp})^2 \Big\},
    \end{align}
    The time derivative of $W^{vp}_i$ in (\ref{eq:Wi_VP}) is:
    \begin{align}\label{eq:dot_Wi_VP_1}
        \nonumber &\dot{W}^{vp}_i = -K_{\Delta p}^{vp}( \Delta p_i - \Delta p_i^r )^2 - K_{\Delta v}^{vp}( \Delta v_i - \Delta v_i^{vp,r} )^2 + \\ 
        \nonumber &\: - \lambda_1\rho_{1,i}^2 - \lambda_2\rho_{2,i}^2 + \rho_{1,i}\rho_{2,i} \\
        \nonumber &\: = -\Tilde{\chi}_i^T \underbrace{ \left[ 
        \begin{array}{cccc}
            q_1 & 2K_{\Delta p}^{vp}K_{\Delta v}^{vp} & q_2 & 2K_{\Delta p}^{vp}K_{\Delta v}^{vp}  \\
            0 & K_{\Delta v}^{vp} & q_3 & 2K_{\Delta v}^{vp} \\
            0 & 0 & q_4 & q_5 \\
            0 & 0 & 0 & \lambda_2+K_{\Delta v}^{vp}
        \end{array}
        \right] }_{Q^{vp}} \Tilde{\chi}_i^T \\
        &\: \leq -\alpha^{vp}|\Tilde{\chi}_i|^2
    \end{align}
    with $q_1 = K_{\Delta p}^{vp}(1+K_{\Delta p}^{vp}K_{\Delta v}^{vp})$, $q_2 = 2K_{\Delta p}^{vp}(1+K_{\Delta v}^{vp}(K_{\Delta p}^{vp}-\lambda_1))$, $q_3 = 2K_{\Delta v}^{vp}(K_{\Delta p}^{vp}-\lambda_1)$, $q_4 = K_{\Delta p}^{vp}+\lambda_1+K_{\Delta v}^{vp}(\lambda_1-K_{\Delta p}^{vp})^2$, $q_5 = 1-2K_{\Delta v}^{vp}(K_{\Delta p}^{vp}-\lambda_1)$, and 
    \begin{align}
        \alpha^{vp} &= \lambda_{\min}(Q^{vp}) = \min\left\{ q_1 , K^{vp}_{\Delta v} , q_4 , \lambda_2+K^{vp}_{\Delta v} \right\}. \label{eq:dot_Wi_alphatilde_VP}  
    \end{align}
    Since $\alpha^{vp} > 0$, then $W^{vp}_i$ is a Lyapunov function for the $i$-th isolated subsystem and exponential stability is proven (see \cite[Theorem~4.10]{B_khalil_2002}).
    
    We move on to prove the ISS property in (\ref{eq:ISS_platoon_property_VP}). Let us consider the function $W^{vp}_i$ for the $i$-th system and its time derivative $\dot{W}^{vp}_i$. In this case, we consider an interconnection term $g_{cl,i}\neq0$. The following inequality is satisfied:
    \begin{align}\label{eq:dot_Wi_VP_2}
        \nonumber &\dot{W}^{vp}_i = -K_{\Delta p}^{vp}( \Delta p_i - \Delta p_i^r )^2 - K_{\Delta p}^{vp}( \Delta v_i - \Delta v_i^{vp,r} )^2 + \\ 
        \nonumber &\: - \lambda_1\rho_{1,i}^2 - \lambda_2\rho_{2,i}^2 + \rho_{1,i}\rho_{2,i} + \rho_{2,i}(a\psi^{i-1}_{\Delta p} + b\psi^{i-1}_{\Delta v}) \\
        &\:\: \leq -\alpha^{vp}|\Tilde{\chi}_i|^2 + |\Tilde{\chi}_i|(a|\psi^{i-1}_{\Delta p}| + b|\psi^{i-1}_{\Delta v}|).
    \end{align}
    Define $d$ as in (\ref{eq:constant_definition_CP}) and $\Upsilon^{vp}\in(0,1)$.
    By Lemmas \ref{lmm:fmacro_bounds} and \ref{lmm:psi_bounds}:
    \begin{align}\label{eq:dot_Wi_VP_3}
        \nonumber \dot{W}_i &\leq -\alpha^{vp}|\Tilde{\chi}_i|^2 + d |\Tilde{\chi}_i|\max_{j=0,...,i-1}|\Tilde{\chi}_j| + \alpha^{vp}\Upsilon^{vp} - \alpha^{vp}\Upsilon^{vp} \\
        &\leq -(1-\Upsilon^{vp})\alpha^{vp}|\Tilde{\chi}_i|^2, \quad \forall \  |\Tilde{\chi}_i|\geq\frac{d}{\alpha^{vp}\Upsilon^{vp}}\max_{j=0,...,i-1}|\Tilde{\chi}_j|.
    \end{align}
    The condition in (\ref{eq:dot_Wi_VP_3}) satisfies the ISS requirements. Consequently, the inequality in (\ref{eq:ISS_platoon_property_VP}) is verified according to \cite[Theorem~4.19]{B_khalil_2002}. Moreover, 
    \begin{equation}\label{eq:gammatilde_VP}
        \gamma^{vp}(s) = \Tilde{\gamma}^{vp}s\ \  \forall \  s \geq 0, \ \ \  \Tilde{\gamma}^{vp} = \sqrt{\frac{\Bar{\alpha}^{vp}}{\underline{\alpha}^{vp}}}\frac{d}{\alpha^{vp}\Upsilon^{vp}} > 0.
    \end{equation}
    Since the parameters $a,b \geq 0$ in the dynamics of $\rho_i$ in (\ref{eq:rho_dynamic_system}) can be arbitrarily selected, the constant $d$ defined in (\ref{eq:constant_definition_CP}) can be chosen such that $\Tilde{\gamma}^{vp}$ in (\ref{eq:gammatilde_VP}) belongs to $(0,1)$.
    
    On the basis of Theorem \ref{thm:StringStability}, Asymptotic String Stability of the platoon can be obtained by using an appropriately chosen function describing macroscopic information.
    
\end{proof}

\end{document}